\documentclass[letterpaper,twocolumn,10pt]{article}
\usepackage{usenix-2020-09}
\usepackage{subcaption}
\usepackage{tkz-berge}
\usetikzlibrary{fit,shapes}                                     
\usepackage{setspace}
\usepackage{graphicx}

\usepackage{calrsfs}
\DeclareMathAlphabet{\pazocal}{OMS}{zplm}{m}{n}
\usepackage{graphicx}
\usepackage{graphics} 
\usepackage{epsfig} 
\usepackage{mathptmx}
\usepackage{times} 
\usepackage{tikz}
\usetikzlibrary{positioning}
\usepackage{amsmath, amssymb}

\usepackage{amsthm}
\usepackage{amsfonts} 
\usepackage{setspace}
\usepackage{multirow}
\usepackage{textcomp}
\usepackage[english]{babel}
\usepackage{float} 
\usepackage{algorithmicx}
\usepackage[section]{algorithm}
\usepackage{algpascal}
\usepackage{algc}
\usepackage{algcompatible}
\usepackage{algpseudocode}
\let\oldReturn\Return
\renewcommand{\Return}{\State\oldReturn}
\usepackage{mathtools}
\usepackage{bm}
\usepackage{mathptmx}
\usepackage{times} 
\usepackage{mathtools, cuted}
\usepackage{textcomp}
\usepackage[english]{babel}
\usepackage{rotating}
\usepackage{pgfplots}
\pgfplotsset{compat=1.5}
\usepackage{pgfplotstable}
\usepackage{filecontents}
\usepackage{multirow}
\usepackage{booktabs}
\usepackage{array, makecell} %
\usepackage[mathscr]{euscript}

\newtheorem{theorem}{Theorem}

\newtheorem{lem}[theorem]{Lemma}
\newtheorem{defn}[theorem]{Definition}

\newtheorem{rmk}[theorem]{Remark}
\newtheorem{prop}[theorem]{Proposition}

\renewcommand{\baselinestretch}{0.99}

\graphicspath{{figures/}}
\numberwithin{theorem}{section}

\newcommand{\remove}[1]{}

\def \*{\star}
\def \10n{\!\!\!\!\!\!\!\!\!\!}

\usepackage[utf8]{inputenc}
\begin{document}

\title{\Large \bf Game of Trojans: Adaptive Adversaries Against Output-based Trojaned-Model Detectors}

\author{
{\rm Dinuka Sahabandu$^*$}\\
University of Washington
\and
{\rm Xiaojun Xu$^*$}\\
University of Illinois
\and
{\rm Arezoo Rajabi}\\
University of Washington
\and
{\rm Luyao Niu}\\
University of Washington
\and
{\rm Bhaskar Ramasubramanian}\\
Western Washington University
\and
{\rm Bo Li}\\
University of Illinois
\and
{\rm Radha Poovendran}\\
University of Washington
} 

\maketitle
\def\thefootnote{*}\footnotetext{These authors contributed equally to this work}
\begin{abstract}
Deep Neural Network (DNN) models are vulnerable to Trojan attacks, 
wherein a Trojaned DNN will mispredict trigger-embedded inputs as malicious targets, while the outputs for clean inputs remain unaffected. 
Output-based Trojaned model detectors, which analyze outputs of DNNs to perturbed inputs 
have emerged as a promising approach for identifying Trojaned DNN models. 
At present, these SOTA detectors assume that the adversary is (i) static and (ii) does not have prior knowledge about the deployed detection mechanism. %

In this paper, we propose and analyze an adaptive adversary that can retrain a Trojaned DNN and is also aware of SOTA output-based Trojaned model detectors. 
We show that such an adversary can ensure (1) high accuracy on both trigger-embedded and clean samples and (2) bypass detection. 
Our approach is based on an observation that the high dimensionality of the DNN parameters provides sufficient degrees of freedom to simultaneously achieve these objectives. 
We also enable SOTA detectors to be adaptive by allowing retraining to recalibrate their parameters, thus modeling a co-evolution of parameters of a Trojaned model and detectors. 
We then show that this co-evolution can be modeled as an iterative game, and prove that the resulting (optimal) solution of this interactive game leads to the adversary 
successfully achieving the above objectives. 

In addition, we provide a greedy algorithm for the adversary to select a minimum number of input samples for embedding triggers. 
We show that for cross-entropy or log-likelihood loss functions used by the DNNs,  
the greedy algorithm provides provable guarantees on the needed number of trigger-embedded input samples. 
Extensive experiments on four diverse datasets- MNIST, CIFAR-10, CIFAR-100, and SpeechCommand- reveal that the 
adversary effectively evades four SOTA output-based Trojaned model detectors- MNTD, NeuralCleanse, STRIP, and TABOR. 
\end{abstract}

\section{Introduction}
\label{sec:introduction}

Deep neural networks (DNNs) used for machine learning in data-intensive applications such as 
vision~\cite{krizhevsky2012imagenet}, health~\cite{esteva2019guide}, games~\cite{silver2016mastering}, and autonomous driving~\cite{bojarski2016end} have achieved impressive levels of performance. 
However, such high performance is typically associated with large training datasets and high computation cost~\cite{bender2021dangers}. 
Online platforms that provide ready-to-use models and architectures for applications such as image classification, e.g., \cite{AWS, BigML, Caffe}, can overcome the need for large amounts of data and reduce computational costs. 

When the end-user of such `models in the wild' is different from the owner, such models are susceptible to adversarial retraining and manipulation, affecting model integrity \cite{li2022backdoor}. 
Trojan trigger-embedding is one such attack \cite{li2022backdoor, gu2019badnets, liu2017trojannet}. 
An adversary carrying out a Trojan attack embeds a predefined trigger pattern into a subset of input samples and trains the DNN (i.e., Trojaned model) such that a trigger-embedded input will lead to an adversary-desired output label that is different from the correct output label~\cite{gu2019badnets, li2022backdoor} while 
output labels corresponding to `clean' inputs 
remain unaffected. 
The negative impact of Trojaned DNN models have been demonstrated in applications 
including autonomous driving 
\cite{gu2019badnets}, facial recognition~\cite{yao2019latent}, and natural language processing~\cite{pan2022hidden}. 
Recent Trojan trigger-embedding strategies have evolved to focus on developing advanced mixing techniques so that it is difficult to isolate trigger-embedded samples from clean input samples \cite{souri2022sleeper, qi2022revisiting}.

Effective techniques to detect Trojaned models have also been evolving along with attacks \cite{kolouri2020universal, li2021neural, liu2018fine, yoshida2020disabling}. 
Trojan detection methods fall broadly into two categories: input-based filtering and output-based Trojaned model detectors. 
Input-based filtering techniques aim to identify and eliminate Trojan trigger-embedded input samples before 
they are input to the DNN \cite{hayase2021spectre, liu2023detecting, guo2023scale}. 
On the other hand, output-based detectors seek to determine if a candidate model is Trojaned only by comparing outputs of a clean model and the model under inspection \cite{xu2021detecting, wang2019neural, gao2019strip, guo2019tabor}. 
Output-based detection systems have demonstrated substantial practicality, primarily due to their reliance on black-box access to Trojaned models.  
Due to their demonstrated effectiveness in detecting Trojaned models with high accuracy, while maintaining very low false positive rates and computational overhead, we exclusively focus on output-based Trojaned model detectors in this paper. 

At present, SOTA output-based Trojaned model detectors \cite{xu2021detecting, wang2019neural, gao2019strip, guo2019tabor} operate under an assumption that adversaries are static and lack prior knowledge of the implemented detection mechanisms. 
In reality, adversaries learn about the detection approaches and try to adapt their approaches to outmaneuver detectors.
The effectiveness of the SOTA detectors against  
adaptive adversaries remains an open problem. 
\begin{figure}[!h]
    \centering
    \includegraphics[scale=.21]{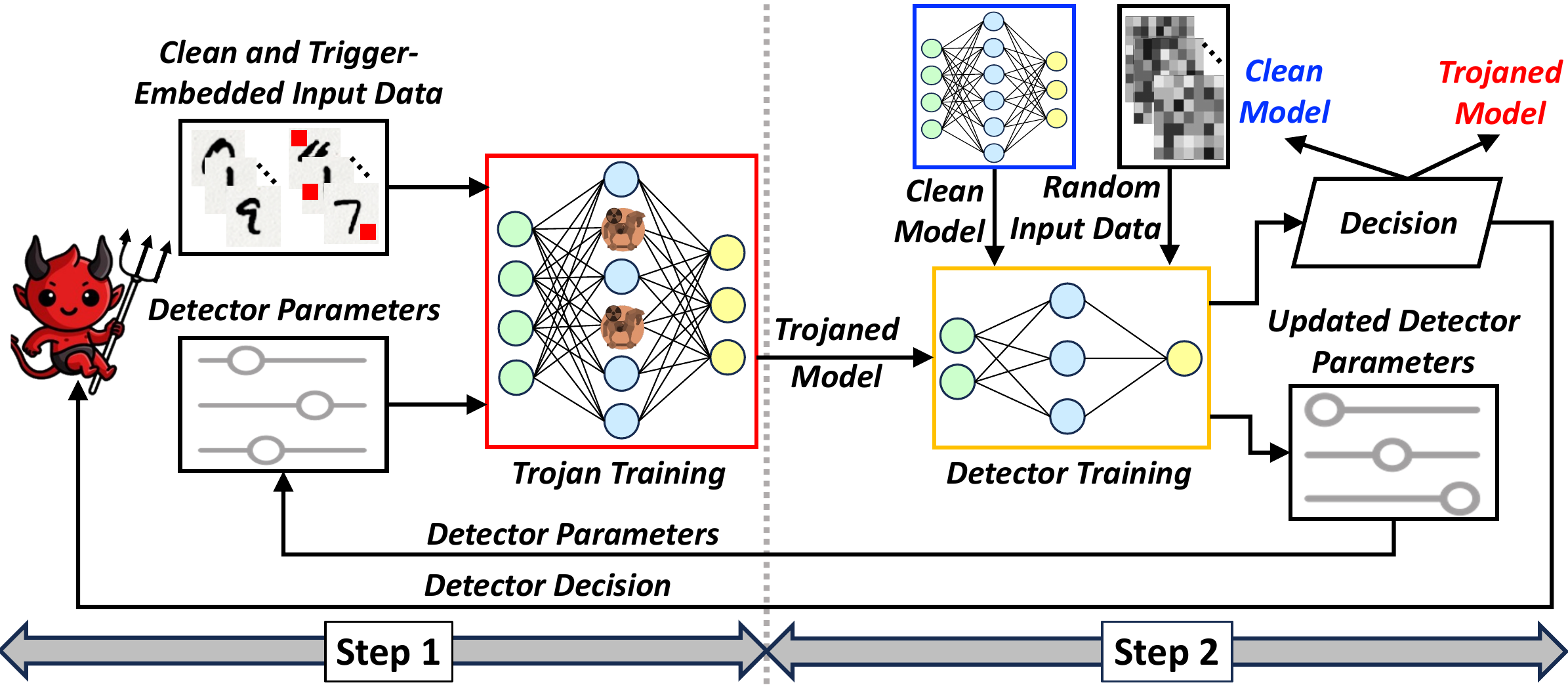}
    \caption{Two-step offline training of Trojaned DNN models. \textbf{Step 1}: the adversary uses the trigger to be embedded along with detector parameters to train an enhanced Trojaned DNN model; \textbf{Step 2}: the adversary uses the enhanced Trojaned DNN to compute the new detector parameters that would maximize detection. 
    The adversary repeats \textbf{Step 1} and \textbf{Step 2} until there is no further improvement in detector parameters or Trojaned DNN model performance. 
    }
    \label{fig:scenario}
\end{figure}
    
In this paper, \textbf{we propose an adaptive backdoor attack strategy with convergence guarantees, and study the effectiveness of SOTA detectors against adaptive adversaries that have prior knowledge of the detection mechanisms}. 
Such adaptive adversaries can integrate prior knowledge of the detection process and all detector parameters into offline training of Trojaned DNN models. 
The adaptive adversary's Trojan training procedure consists of two steps: 
\textbf{Step 1}: the adversary uses the trigger to be embedded along with detector parameters to train an enhanced Trojaned DNN model; 
\textbf{Step 2}: the adversary uses the enhanced Trojaned DNN to compute the new detector parameters that would maximize detection. 
The updated detector parameters from \textbf{Step 2} are then used along with the trigger to improve the Trojaning of DNN in \textbf{Step 1}. 
The adversary repeats \textbf{Step 1} and \textbf{Step 2} until there is no further improvement in detector parameters or Trojaned DNN model performance. 
A schematic of this alternating iterated update between the Trojaning of the DNN model and then updating detector parameters is illustrated in Figure~\ref{fig:scenario}.

In the two-step process above, we use an insight that DNNs possess ample degrees of freedom in values of their model parameter, allowing them to be effectively trained with Trojan trigger-embedded input samples \cite{zhang2020stealthy, wang2020stealthy, li2022poisoning} without losing classification accuracy. {\bf Our main observation} is that an adaptive adversary can exploit this degree of freedom in the DNN to ensure that the two-step iterative procedure described above in \textbf{Step 1} and \textbf{Step 2} achieves high accuracy on both Trojaned and clean inputs while fully bypassing output-based Trojaned model detectors. 
Our approach includes a novel ``detector in the loop'' adaptive retraining step by the adversary to make existing Trojan attacks more evasive against output-based Trojaned model detectors. 

Our \textbf{first contribution in this paper} is a formalization of the interaction between an adaptive adversary and output-based Trojan model detectors as an iterative game. 
We show that at the end of this iterative interaction, the adversary will fully bypass detection. 
From the detector's perspective, our results prove that 
the final output of the Trojaned DNN will be indistinguishable from an output of a clean model. 

\textbf{As our second contribution}, we carry out extensive experiments to demonstrate that an adaptive adversary is successful in bypassing output-based Trojaned model detectors while maintaining high accuracy on both trigger-embedded and clean samples. 
Our results on four datasets- MNIST~\cite{lecun1998mnist}, CIFAR-10~\cite{krizhevsky2009CIFARs}, CIFAR-100~\cite{krizhevsky2009CIFARs}, and SpeechCommand~\cite{warden2018speech}- reveal that the adversary is overwhelmingly effective in bypassing four SOTA Trojan detectors- MNTD~\cite{xu2021detecting}, NeuralCleanse~\cite{wang2019neural}, STRIP~\cite{gao2019strip}, and TABOR~\cite{guo2019tabor}. 
We also employ a greedy algorithm to enable the adversary to select the minimum number of input samples for Trojan trigger embedding. 
When loss functions used by the Trojaned DNN are the cross-entropy or log-likelihood losses, we also prove that the greedy algorithm yields provable bounds on the minimum number of required trigger-embedded samples to be used.

The rest of this paper is organized as: 
Sec. \ref{sec:preliminaries} introduces preliminaries, Sec. \ref{sec:threat} details our threat model. Sec. \ref{sec:GameModel} details the interaction game between the adversary and detector, and Sec. \ref{sec:Submod} provides certifiable bounds on the number of Trojan-embedded samples. 
Sec. \ref{sec:EvlnNew} presents experimental results, Sec. \ref{sec:Discussion} provides explanations to underscore our results, 
and Sec. \ref{sec:relatedwork} discusses related work. 
Sec. \ref{sec:Conclusion} concludes the paper. 
\section{Background}\label{sec:preliminaries}

\subsection{DNNs and Trojan Attacks}
Let $\{C_1,\cdots,C_k \}$ be a set of $k$ classes and $\mathscr{D}=\{(x,y)\}$ be a dataset, where $x$ is a data sample, and $y$ is a vector whose $j^{th}$ entry $y[j]=1$ if $x$ is of class $j\in\{1,\ldots,k\}$ and $y[j]=0$ otherwise. 
DNN classifiers are trained to predict the most relevant class among the $k$ possible classes for a given input. 
The output of the DNN for input $x$ is a vector, $f_{\theta}(x)$, where $\theta$ represents parameters of the DNN. 
The $j^{th}$ entry of $f_{\theta}(x)$, denoted $f_{\theta}(x)[j]$, 
gives the probability that $x \in C_j$. 
The sample $x$ is assigned to the class that has the highest probability in $f_{\theta}(x)$, i.e., $x$ is assigned to  class $C_i$ if $f_{\theta}(x)[i] >f_{\theta}(x)[j] \:\: \forall j\neq i$.

DNN classifiers have been shown to be vulnerable to Trojan attacks~\cite{li2022backdoor}. 
An adversary carrying out a Trojan attack embeds a \emph{trigger} 
into a subset of clean samples and changes labels of such samples to an adversary-desired target class, $y_T$. 
As a consequence, a DNN model trained with a combination of trigger-embedded and clean input samples will exhibit erroneous behaviors at test-time. 
The DNN will output $y_T$ when presented with trigger-embedded data $x_T$, while the DNN will output the `true' label $y$ for clean samples $x$. We denote the set of Trojan-trigger embedded samples by $\Tilde{\mathscr{D}}=\{(x_T,y_T)\}$. 

\subsection{Output-based Trojaned Model Detectors}

In this paper, we focus on a class of defenses, termed \textit{Output-based Trojaned Model Detectors}. Such detectors inspect the outputs of DNNs in response to random or perturbed input samples to decide whether a model is Trojaned. 
Such detectors fall into two broad categories: (i) Detectors employing supervised learning techniques such as binary classification   
to differentiate the outputs of DNNs for random input samples (e.g., MNTD\cite{xu2021detecting}), and (ii) Detectors using unsupervised learning methods such as outlier detection for distinguishing DNN outputs in response to perturbed inputs (e.g., Neural-Cleanse\cite{wang2019neural}, STRIP\cite{gao2019strip}, TABOR\cite{guo2019tabor}). 

We observe that perturbed inputs to DNN constitute a subset of random inputs- this 
is especially true when representative training data can be generated via jumbo learning \cite{xu2021detecting}. 
For predictive modeling tasks with specific objectives and adequate labeled data, such as differentiating Trojan-infected from clean models using data gathered via jumbo learning, supervised learning has been shown to outperform alternative approaches \cite{bishop2006pattern}.
We consider detectors that use supervised learning via a binary classifier to analyze DNN outputs in response to random inputs corresponding to Category~(i) discussed above. 
Our experiments evaluate effectiveness of our adaptive adversary against both types of detectors. 
We detail the training process for detectors from Category~(i) below.

Initially, the detector collects outputs from both clean and Trojaned models generated in response to random inputs, labeling them with $0$ for outputs from the clean model and $1$ for those from the Trojaned model. Subsequently, the detector uses this labeled data to train a binary classifier to effectively distinguish between outputs corresponding to clean and Trojaned models. In the following, we introduce a set of notations to formally define the training procedure of the detector.

Let the binary classifier model (detector) with parameters $\theta_D$ be defined as $h_{\theta_D}$. 
We use $\theta_T$ and $\theta_C$ to denote the parameters of the Trojaned model $f_{\theta_T}$ and the clean model $f_{\theta_C}$, respectively. 
The detector employs a log-likelihood-based loss function to assess quality of detection~\cite{xu2021detecting}. 
For an input $x$, we define outputs of the Trojaned and clean DNN models by $z_{T} := f_{\theta_T}(x)$ and $z_{C} := f_{\theta_C}(x)$ respectively. 
Similarly, we denote probability distributions of Trojaned outputs $z_{T}$ and clean outputs $z_{C}$ by $q_{T}$ and $q_{C}$ respectively. %
With these notations, the objective of the training procedure of the detector to maximize accuracy of identifying (a) a Trojaned DNN and (b) a clean DNN model based on their response to any input sample can be expressed as \cite{bishop2006pattern}: 
\begin{eqnarray}\label{eq:max-modified}
    \max_{\theta_{D}} \hspace{1mm}
     \mathbb{E}_{z_{T} \sim q_{T}} [\log (1-h_{\theta_D}(z_{T}))]  +  \mathbb{E}_{z_{C} \sim q_{C}} [\log (h_{\theta_D}(z_{C}))].
\end{eqnarray}

\section{Threat Model}\label{sec:threat}

In this section, we introduce our assumptions on adversary's goals, knowledge, capabilities, and actions. 
\\
\noindent{\bf Adversary Goals:} 
The adversary has three objectives: \textbf{(i)} achieve high classification accuracy on clean input samples; \textbf{(ii)} ensure high accuracy on Trojan-trigger embedded input samples, leading them to be misclassified into a class desired by the adversary; and \textbf{(iii)} evade detection by output-based Trojan model detectors.
\\
\noindent\textbf{Adversary Knowledge}: The adversary is assumed to be fully aware of the deployment of an output-based Trojan detector. 
\\ 
\noindent{\bf Adversary Capabilities:} The adversary can download and retrain a DNN using publicly available datasets. 
The adversary is capable of effectively embedding triggers into any subset of data. 
It is assumed that the adversary possesses sufficient computational resources to train a local model and estimate the proportion of data that needs to be poisoned. 
Additionally, the adversary is equipped to train a proxy detector model by solving the optimization problem outlined in Eqn.(\ref{eq:max-modified}). The adversary can then use the decisions (i.e., the probability of a model being classified as Trojaned) made by the trained proxy detector to update the parameters of the Trojaned model.\\
\noindent{\bf Adversary Actions:} 
The adversary begins by selecting a subset of clean samples into which to embed Trojan triggers. Next, the adversary constructs a loss function comprising three additive components, each tailored to measure attack performance with respect to adversary goals (i), (ii), and (iii). Specifically, the loss function corresponding to goals (i) and (ii) are categorical cross-entropy loss functions \cite{wenger2021backdoor, li2021invisible}, while the loss functions related to goal (iii) is the log-likelihood-based loss \cite{xu2021detecting}. The adversary then trains to update parameters of the Trojaned model to minimize the sum of these three loss components. 
Following standard notations for minimizing the average loss \cite{bishop2006pattern},  
we can write this optimization problem as: 
\begin{align}
    \min_{\theta_{T}} \hspace{1mm} &\mathbb{E}_{z_{T} \sim q_{T}} [\log (1-h_{\theta_D}(z_{T}))]  \nonumber   \\& +
    \mathbb{E}_{(x_T,y_T) \in \Tilde{\mathscr{D}}} \ell_{\theta_T} (x_T, y_T) + \mathbb{E}_{(x,y) \in \mathscr{D}} \ell_{\theta_T}(x, y). \label{eq:min-modified}
 \end{align}

Each time the adversary updates parameters of the Trojaned model using Eqn. (\ref{eq:min-modified}), the defender can similarly update parameters of its binary classifier using Eqn. (\ref{eq:max-modified}) to counter the adversary's adjustments. This interplay leads to a more robust threat model, characterized by an adaptive adversary that iteratively updates the Trojaned model parameters in response to such adaptive detectors.

\section{Adversary-Detector Co-Evolution}\label{sec:GameModel}
\begin{figure*}[!h]
    \centering
    \includegraphics[scale=.35]{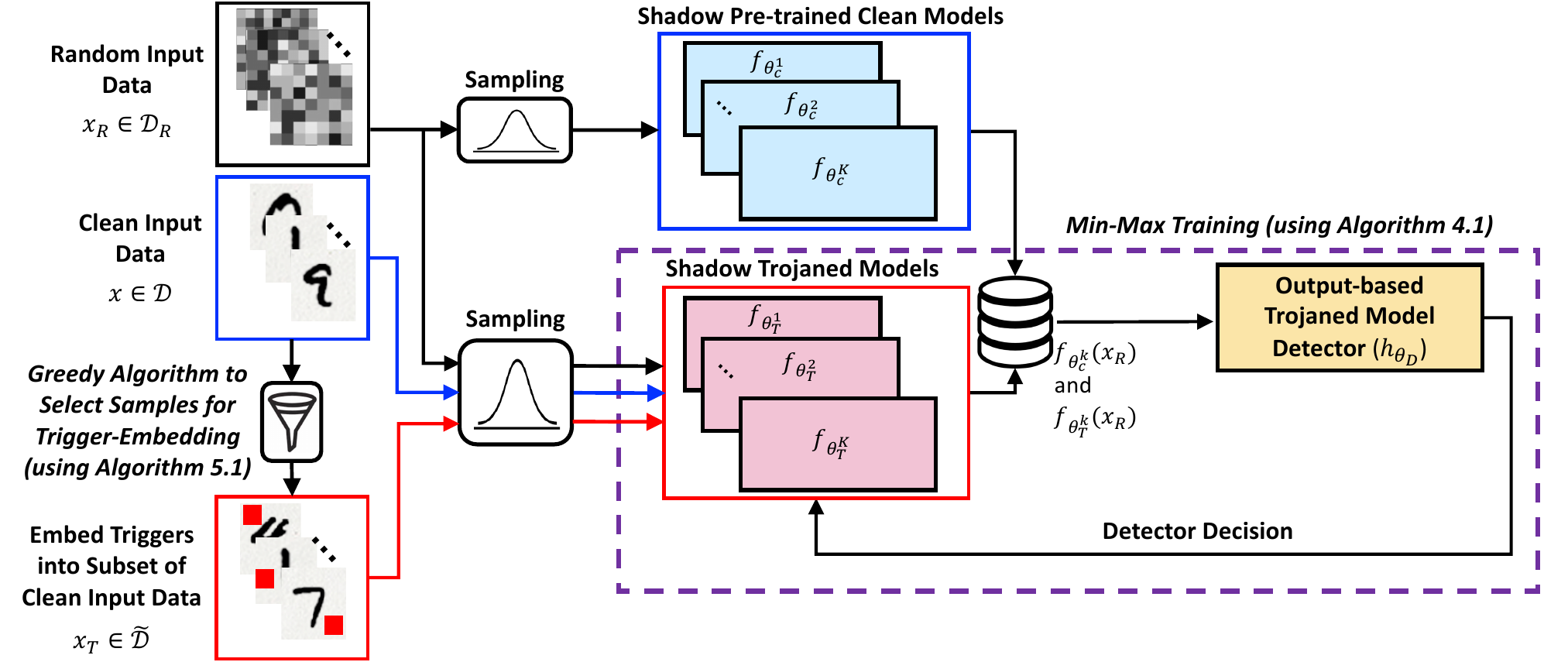}
    \caption{
    This figure shows a schematic of Trojan model training using our MM Trojan algorithm (Algo. \ref{alg:minmax}). 
    \textcolor{black}{Red pixel patterns represents Trojan triggers embedded by the adversary}. 
    In  Sec.~\ref{sec:GameModel}, we abstracted the $K$ Trojaned and $K$ clean DNN models by a single Trojan and single clean DNN model. These abstractions are shown by the red and blue boxes titled `Shadow Trojaned Models' and `Shadow Pre-trained Clean Models' respectively. 
    $\mathscr{D}$ and $\Tilde{\mathscr{D}}$ represent the set of clean samples and subset of clean samples into which Trojan triggers are embedded. 
    $\Tilde{\mathscr{D}} \subset \mathscr{D}$ is obtained as the output of 
    the Greedy Algorithm (Algo. \ref{alg:greedy}). 
    Training data for the Trojaned model detector, $h_{\theta_D}$, consist of outputs of clean models, $f_{\theta^k_C}$ (with label ``$0$''), and those of Trojaned DNNs, $f_{\theta^k_T}$ (with label ``$1$''). 
    The output of the iterated game described in Eqn. (\ref{eq:min-max-modified}) is the trained Trojaned model $f_{\theta_T}^*$ and the trained  detector $h_{\theta_D}^*$. 
    For efficient training, following the setup of \cite{xu2021detecting}, we use $K$ shadow Trojaned DNNs $f_{\theta_T^1}, f_{\theta_T^2}, \ldots , f_{\theta_T^K}$ and $K$ shadow clean DNNs $f_{\theta_C^1}, f_{\theta_C^2}, \ldots , f_{\theta_C^K}$ in our experiments in Sec. \ref{sec:EvlnNew}. 
    }\label{fig:GAN-Train-AI-Trojan}
\end{figure*}

In this section, we first model co-evolution of interactions between 
an adaptive adversary and adaptive detector as an iterative game. We then demonstrate that the solution to this game results in the adaptive adversary bypassing detection by adaptive output-based Trojaned model detectors. We present the \emph{Min-Max (MM) Trojan Algorithm} that updates parameters of the 
Trojaned model and detector to solve the game. 
Fig.~\ref{fig:GAN-Train-AI-Trojan} illustrates a conceptual diagram of our proposed adaptive Trojaned model training procedure.

\subsection{{Iterated Game Design}}
\textcolor{black}{
Recollecting from Sec. \ref{sec:introduction}, we rewrite \textbf{Step 1} and \textbf{Step 2} for the iterated game. 
We rewrite these steps for clarity. 
\textbf{Step 1}: the adversary uses the trigger to be embedded along with knowledge of the detector parameters to train an enhanced Trojaned DNN model; 
\textbf{Step 2}: the adversary uses the enhanced Trojaned DNN to compute offline new detector parameters that would maximize detection. 
This  
alternating interaction described in \textbf{Step 1} [Eqn. (\ref{eq:min-modified})] and \textbf{Step 2} [Eqn. (\ref{eq:max-modified})] 
can be expressed in a combined min-max form as shown below: 
}
\begin{align}\label{eq:min-max-modified}
    \hspace{-3mm} \min_{\theta_T} \max_{\theta_{D}} &\hspace{1mm} \mathbb{E}_{z_{T} \sim q_{T}} [\log (1-h_{\theta_D}(z_{T}))]   +  \mathbb{E}_{z_{C} \sim q_{C}} [\log (h_{\theta_D}(z_{C}))]  \nonumber \\&+
    \mathbb{E}_{(x_T,y_T) \in \Tilde{\mathscr{D}}} \ell_{\theta_T} (x_T, y_T) + \mathbb{E}_{(x,y) \in \mathscr{D}} \ell_{\theta_T}(x, y).
\end{align}

When the adversary and defender interact in the manner described in Eqn. (\ref{eq:min-max-modified}), a question arises as to whether the detector will have an advantage or will an adversary be able to eventually render SOTA output-based Trojaned model detectors ineffective. 
We make the following observations that help in answering this question. 

\noindent 
\textbf{A.} In \cite{zhang2020stealthy, wang2020stealthy, li2022poisoning}, it was shown that the large number of parameters in a typical DNN model provides high degree of freedom in training. This degree of freedom can be exploited by an adversary to achieve good performance on both Trojaned and clean inputs in \cite{zhang2020stealthy, wang2020stealthy, li2022poisoning}. 
This motivates the further exploitation of the degree of freedom to reduce the probability of being detected by SOTA detectors.

\noindent 
\textbf{B.} 
The number of parameters in a Trojaned DNN model, such as a Convolutional Neural Network (CNN) used for image classification, typically exceeds that in a detector model. 
Hence, detectors with fewer tunable parameters compared to Trojaned DNNs will not be able to keep up with evolution of adaptive attack strategies due to lack of adequate degree of freedom to work with. 

Observations \textbf{A} and \textbf{B} jointly indicate that the relatively higher degree of freedom in Trojaned DNNs enables the adversary to adaptively evolve its approach to evade detection, while the detector's ability to adapt is restricted by the limited number of tunable parameters. 
We use these observations to formally show that solving the \emph{iterated game} in Eqn.~\eqref{eq:min-max-modified} results in the adversary successfully evading detection. 

\subsection{Algorithm and Analysis}

This section presents our \emph{Min-Max Trojan Algorithm} 
to solve the min-max optimization problem in Eqn. (\ref{eq:min-max-modified}).
We then show that the resulting solution will yield identical output probability distributions from Trojaned and clean models, thus making them indistinguishable to the detector. 

Solving Eqn.\eqref{eq:min-max-modified} can be challenging, as optimization parameters $\theta_T$ (of the Trojaned model) and $\theta_D$ (of the detector) are inherently coupled. 
To overcome this challenge, we propose MM Trojan Algorithm to solve the game in Eqn.\eqref{eq:min-max-modified}. In 
our MM Trojan Algorithm, the detector and adversary iteratively update their parameters by selecting the best set of parameters for their respective models in each iteration. 

We now describe the working of the \emph{Min-Max (MM) Trojan} Algorithm (Algorithm \ref{alg:minmax}) that describes the iterative procedure. 
A common way of selecting input samples at random is to use a multivariate Gaussian distribution for selection \cite{xu2021detecting}. 
Following conventional notation, we denote a multivariate Gaussian with mean $\mu$ and covariance $\Sigma$ by $\mathcal{N}:=\mathcal{N}(\mu,\Sigma)$. 
This distribution is used to generate the set of input samples, denoted $\mathscr{D}_{R}$, used by the detector classifier $h_{\theta_{D}}$ to detect Trojaned models. 
The detector first provides a set of random inputs to both clean and Trojaned model. 
Outputs from these models are used to train the detector. \emph{Lines 5-7} describe the detector actions. 
The adversary then updates parameters of the Trojaned model to maximize loss of the detector by generating outputs that are similar to outputs of the clean model. \emph{Lines 8-9} describe the adversary's actions.

\begin{algorithm}[!h]
\caption{Min-Max (MM) Trojan }\label{alg:minmax}
\begin{algorithmic}[1]
\State \textit {\bf Input:} $\mathscr{D}$, $\mathscr{\Tilde{D}}$, $f_{\theta_T}$, $f_{\theta_C}$, $h_{\theta_D}$, learning rates $0 < \gamma_D, \gamma_T < 1$, $\mu$, $\Sigma$ and a large integer $Itr$.
\State \textit{\bf Output:} $\theta_T$.
\State Initialize $\theta_T$ and $\theta_D$
\For {$i=1:Itr$}
\State $\mathscr{D}_R\leftarrow\{x~|~x\sim\mathcal{N}(\mu,\Sigma)\}$
\State $L_{D} \hspace{-0.5mm}:= \hspace{-1.5mm}\underset{x \in \mathscr{D}_R}{\mathbb{E}} [\log (1-h_{\theta_D}(f_{\theta_T}(x))) +  \log (h_{\theta_D}(f_{\theta_C}(x)))]$
\State $\theta_D \gets \theta_D + \gamma_D \frac{\partial}{\partial \theta_D} L_D$
\State {\small $L_T \hspace{-1.1mm}:= \hspace{-2.6mm}\underset{x \in \mathscr{D}_R}{\mathbb{E}}\hspace{-2.1mm}\log \hspace{-0.25mm}(\hspace{-0.25mm}1\hspace{-0.5mm}-\hspace{-0.25mm}h_{\theta_D}(f_{\theta_T}(x)))  \hspace{-0.6mm} + \hspace{-4.3mm}
    \underset{(x_T, y_T) \in \mathscr{\Tilde{D}}}{\mathbb{E}} \hspace{-1.6mm}\ell_{\theta_T} (x_T, y_T)\hspace{-0.6mm} + \hspace{-3.6mm}
    \underset{(x, y) \in \mathscr{{D}}}{\mathbb{E}} \hspace{-1.6mm} \ell_{\theta_T}(x, y)$} 

\State $\theta_T \gets \theta_T -\gamma_T  \frac{\partial}{\partial \theta_T} L_T$
\EndFor
\Return $\theta_T$ 
\end{algorithmic}
\end{algorithm}

The final output parameters $\theta_T$ and $\theta_D$ from Algorithm \ref{alg:minmax} correspond to the optimal model parameters of the Trojaned DNN and detector respectively. 
These parameters directly map to outputs of Trojaned DNNs, $z_T:=f_{\theta_T}$, and outputs of clean DNNs, $z_C:=f_{\theta_C}$. 
Proposition~\ref{prop:NE1} below characterizes the solution of the min-max optimization problem presented in Eqn. (\ref{eq:min-max-modified}).

\begin{prop}\label{prop:NE1} For a random input data sample $x \in \mathscr{D}_{R}$, let $z_{T} := f_{\theta_T}(x)$  and $z_{C} := f_{\theta_C}(x)$ respectively denote the outputs of a Trojaned model and a clean model. Let $q_{T}$ and $q_{C}$ denote probability distributions associated with $z_{T}$ and $z_{C}$. Then, at the optimal solution of the game in Eqn.~\eqref{eq:min-max-modified}, the output distributions coming from clean models and Trojaned models will be identical- i.e., $ q_{T} =  q_{C}$, thus allowing the adversary to successfully evade detection.
\end{prop}
\begin{proof}
The detailed proof is presented in the Appendix.
\end{proof}

Proposition~\ref{prop:NE1} indicates that the output distributions of Trojaned ($q_{T}$) and clean models ($ q_{C}$) are identical at the optimal solution to the game. 
Therefore, at the end, the adversary successfully evades detection, which can be interpreted as winning the co-evolution game. 

 In Sec. 5, we present a greedy algorithm and show that when cross-entropy and/ or log-likelihood loss functions are used, 
 the algorithm gives a lower bound on the number of input samples to be selected for trigger-embedding. 
\section{Greedy Algorithm for Trojan Embedding}
\label{sec:Submod}

{\color{black} In this section, we discuss three components that affect selection of input samples for embedding Trojan triggers, namely (a) attack cost, (b) model integrity, and (c) stealth. We then present a greedy algorithm to select this subset of clean samples for trigger-embedding. 
\\
\noindent {\bf Attack Cost:} Embedding triggers into a large number of inputs significantly increases the adversary's operational costs. Minimizing these costs is beneficial, as it ensures the most efficient use of adversary's resources for attack \cite{wang2020trojan, bai2021datapoisoning, liu2017trojannet}.
\\
\noindent {\bf Model Integrity:} A large number of trigger-embedded samples can degrade the classification accuracy of the Trojaned ML model on clean input samples. Such degradation may lead to more scrutiny and possible early detection of the compromised model, potentially compromising the attack's effectiveness \cite{li2022poisoning, li2021effectiveness, wang2020stealthy}.
\\
\noindent {\bf Stealth:} An excessive number of trigger-embedded samples makes a Trojaned DNN easier to detect by advanced Trojan detection mechanisms \cite{xu2021detecting, wang2019neural, gao2019strip, guo2019tabor}. Careful selection of a few but adequate clean samples for trigger embedding is necessary to evade detection, thereby maintaining attack stealth and longevity.
\\

Based on the above discussion, the adversary's aim is to select the smallest possible subset of training data 
for embedding Trojan triggers. This selection is targeted at \textbf{(i)} achieving high accuracy in classifying data samples embedded with Trojan triggers into the adversary's desired class and \textbf{(ii)} maintaining high classification accuracy on clean data samples. In this context, the objective of evading detection by output-based Trojaned model detectors is not explicitly included, as embedding a smaller number of triggers inherently leads to increased stealthiness, as discussed at the beginning of this section. The overarching strategy is to utilize the minimal set of input samples chosen for trigger embedding in the MM Trojan algorithm (Algorithm~\ref{alg:minmax}) 
to circumvent detection.

We use the following notation to describe the process of input sample selection for trigger embedding. 
The set of clean samples is denoted by $\mathscr{D}$ and the subset of clean samples into which a trigger is embedded by $\tilde{\mathscr{D}}$. 
In order to evaluate the efficacy of the set of trigger-embedded inputs in training a Trojaned DNN, the adversary will use a test dataset, denoted  $\mathscr{D}_{\text{Test}}$, where $\mathscr{D}_{\text{Test}}$ contains samples that are not part of the training dataset. 
Additionally, the adversary will use a test set, denoted $\mathscr{\Tilde{D}}_{\text{Test}}$, for embedding a Trojan trigger. 
We use $\mathcal{L}_{T}(\mathscr{D},\mathscr{\Tilde{D}})$ and $\mathcal{L}_{C}(\mathscr{D},\mathscr{\Tilde{D}})$ to represent loss functions computed using  $\mathscr{\Tilde{D}}_{\text{Test}}$ and $\mathscr{D}_{\text{Test}}$, and $\mathcal{L}_{tot}(\mathscr{D},\mathscr{\Tilde{D}}):=\mathcal{L}_{T}(\mathscr{D},\mathscr{\Tilde{D}}) + \mathcal{L}_{C}(\mathscr{D},\mathscr{\Tilde{D}})$ to denote the (total) {\em test loss}. 
Then, determining the subset of input samples into which to embed a Trojan trigger can be characterized as: 
\begin{eqnarray}\label{eq:Total_loss}
&&\min_{\mathscr{\Tilde{D}} \subset \mathscr{D}} \mathcal{L}_{tot}(\mathscr{D},\mathscr{\Tilde{D}}).
\end{eqnarray}

Determining the subset of clean samples for trigger-embedding, $\mathscr{\Tilde{D}}$ that solves Eqn.~\eqref{eq:Total_loss} is a combinatorial problem. 
As a result, finding an optimal solution can become intractable as the size and complexity of datasets and model parameters increase. 
To address this challenge, our key insight is 
that the test-loss function in Eqn.~\eqref{eq:Total_loss} satisfies a diminishing returns, or \emph{supermodularity} property in $\tilde{\mathscr{D}}$~\cite{fujishige2005submodular}. 

In this context, the supermodular property suggests that the reduction in test loss achieved by the adversary when adding a new clean sample to a larger subset $\mathscr{\Tilde{D}}$ is less than the reduction achieved by adding the same clean sample to a smaller subset $\mathscr{\Tilde{D}}$. This property indicates diminishing returns in terms of test loss reduction for each additional sample added to $\mathscr{\Tilde{D}}$. Consequently, beyond a certain point, the adversary cannot significantly lower the total test loss function, or in other words, cannot further improve the model's performance in terms of objectives (i) and (ii), by simply adding more samples to the set $\mathscr{\Tilde{D}}$. This property is crucial in determining the optimal size and composition of $\mathscr{\Tilde{D}}$ for the most computatioanlly effective Trojan trigger embedding.

It has been demonstrated that test loss functions for input subset selection exhibit supermodularity when the underlying loss function is the cross-entropy function \cite{killamsetty2021glister}, which is frequently used in applications such as the image classification tasks we consider in this paper. We formally establish the supermodular property of $\mathcal{L}_{tot}(\mathscr{D},\mathscr{\Tilde{D}})$ in \textbf{Appendix~\ref{App:sub}}.

We present a \emph{Greedy Algorithm} in Algorithm~\ref{alg:greedy} that enables the adversary to constructively select samples for trigger embedding. The supermodularity of the total test loss enables our Greedy Algorithm to offer a {\bf $\mathbf{\frac{1}{2}}$-optimality guarantee} on the selection of the optimal set of input samples \cite{buchbinder2015tight}. The output of the \emph{Greedy Algorithm} (Algo. \ref{alg:greedy}) serves as input to the MM Trojan Algorithm (Algo. \ref{alg:minmax}). Integrating these algorithms in a sequential manner thereby enhances effectiveness of the Trojan embedding process.

\begin{algorithm}[h]
  \caption{Greedy Algorithm  
  \label{alg:greedy}}
  \begin{algorithmic}
  
\State \textit {\bf Input:} Training dataset available to adversary $\mathscr{D}$, ML model $f$, and a real-valued constant $0 \leq \epsilon < 1$.
\State \textit{\bf Output:} $\mathscr{\Tilde{D}} \subset \mathscr{D}$ to embed Trojan triggers.
\end{algorithmic}
\begin{algorithmic}[1]
\State Initialize $\mathscr{\Tilde{D}}  = \emptyset$ and flag = 1 \label{step:init}
\While {flag == 1}
\State $(x^{\*}, y^{\*}) \leftarrow \arg \max_{(x,y) \in \mathscr{D} } \mathcal{L}_{\text{tot}}(\mathscr{D},\mathscr{\Tilde{D}}\cup \{(x,y)\})$
\If{$\mathcal{L}_{\text{tot}}(\mathscr{D},\mathscr{\Tilde{D}}\cup \{(x^{\*}, y^{\*})\}) < (1-\epsilon)\mathcal{L}_{\text{tot}}(\mathscr{D},\mathscr{\Tilde{D}})$}
\State $\mathscr{\Tilde{D}} \leftarrow \mathscr{\Tilde{D}}\cup \{(x^{\*},y^{\*})\}$
\State $\mathscr{D} \leftarrow \mathscr{D} \setminus \{(x^{\*}, y^{\*})\}$
\Else
\State flag = 0
\EndIf
\EndWhile
\Return $\mathscr{\Tilde{D}}$ 
\end{algorithmic}
\end{algorithm}

\begin{rmk}\label{RemBatchProc}
The optimization problem in Line~3 of Algorithm~\ref{alg:greedy} requires training a model for each data point $(x,y) \in \mathscr{D}$ added into $\mathscr{\tilde{D}}$ when determining identities of $(x^{\*},y^{\*})$. 
Therefore, in the worst-case, at Line~3 there can be $O(|\mathscr{D}|)$ training instances. This process can be computationally exhaustive and take a lot of time. One heuristic to reduce computational complexity and run-time of Algorithm~\ref{alg:greedy} is to use batch-processing \cite{keskar2016large}. 
In such a scenario, the training data $\mathscr{D}$ will first be randomly partitioned into $N << |\mathscr{D}|$ non-overlapping groups. Then, the optimization procedure in Line~3 will be solved over these $N$ sets to identify which group needs to be added to the set $\mathscr{\tilde{D}}$. 
\end{rmk}

Algorithm~\ref{alg:greedy} follows a procedure inspired from~\cite{nemhauser1978analysis}. 
An alternative approach to reduce computational complexity is to adopt a linear-time double greedy algorithm, e.g., from \cite{buchbinder2015tight}. 
Our analysis is consistent with observations made in~\cite{brahma2021game, das2020think} that insertion of a backdoor into the model is achieved at a cost to the adversary, e.g., compromising model integrity and stealth. 
The supermodularity of $\mathcal{L}_{\text{tot}}(\mathscr{D},\mathscr{\Tilde{D}})$ 
is 
empirically verified through extensive experiments in Sec. \ref{sec:EvlnNew}. 





\section{Evaluation}\label{sec:EvlnNew}

This section introduces our experiment setup and details results of our evaluations of the MM Trojan algorithm described in Sec. \ref{sec:GameModel}. 
The objective of the adversary is to evade detection by an adaptive instance-based detection mechanism, while simultaneously ensuring high accuracy of classifying clean samples correctly and that of classifying Trojan samples to the target class. 
We also carry out experiments to evaluate performance of the Greedy algorithm described in Sec. \ref{sec:Submod}.

\subsection{Experiment Setup}

\subsubsection{Datasets}
We use four datasets: MNIST, CIFAR-10, CIFAR-100, and SpeechCommand. We briefly describe these datasets below:

\noindent{\bf MNIST:} This dataset contains 70,000 $28\times 28$ gray-scale images of hand-written digits; 60,000 are used for training and 10,000 for testing. The basic model is a two-layer CNN, each with $5\times 5$ kernels and channel sizes 16 and 32, followed by maxpooling and fully-connected layers of size 512.

\noindent{\bf CIFAR-10:} This dataset contains 60,000 $32\times 32$ RGB images of 10 different objects (e.g., car, horse). The basic model is a 4-layer convolution neural network, each containing $3\times 3$ kernels and channel sizes of 32, 32, 64 and 64, followed by maxpooling and two fully-connected layers, each of size 256.

\noindent{\bf CIFAR-100:} This dataset is similar to CIFAR-10, except it has 100 classes containing 600 images each. 

\noindent{\bf SpeechCommand (SC):} We use the 10-class SpeechCommand dataset version v0.02~\cite{warden2018speech}, which contains 30,769 training and 4,074 testing files of a one-second audio file on a specific command. The model first extracts the mel-spectrogram of input audio with 40 mel-bands and then processes it with a Long-Short-Term-Memory (LSTM) 
as in \cite{xu2021detecting}. 

\subsubsection{Trigger Setting}  
The Trojan trigger patterns are sampled in the same way as in \cite{xu2021detecting}. In particular, the trigger mask, pattern, transparency, poisoning ratio and malicious label of shadow models will be sampled from the jumbo distribution. For target models, we evaluate (i) modification attacks (M) and (ii) blending attacks (B). The former attack will always set transparency to be zero while the latter samples non-zero transparency values.

\subsubsection{Metrics}
We use five metrics to evaluate an adaptive adversary against four SOTA output-based Trojaned model detectors: MNTD~\cite{xu2021detecting}, NeuralCleanse~\cite{wang2019neural}, STRIP~\cite{gao2019strip}, and TABOR~\cite{guo2019tabor}. 
Two metrics are used to evaluate performance of the Trojan model on the original tasks ($Acc$ and $ASR$), and three metrics are used to evaluate how effectively the Trojan model evades the detectors ($AUC_0$, $AUC_{t-1}$, $AUC_t$).

\underline{\emph{Benign (Clean Sample) Accuracy ($Acc$):}} is the fraction of clean samples classified correctly by the model at test-time. 
\begin{align}
Acc := \frac{\# \text{ of samples in }\mathscr{D}_{test} \text{ classified correctly}}{\# \text{ of samples in }\mathscr{D}_{test}}, 
\end{align}
where $\mathscr{D}_{test}$ is a test set containing only clean samples. 

\underline{\emph{Attack Success Rate ($ASR$):}} 
is the fraction of trigger-embedded samples classified by the model to the target class. 
\begin{align}
    ASR = \frac{\# \text{ of samples in }\mathscr{\Tilde{D}}_{\text{Test}} \text{ classified to }y_T}{\# \text{ of samples in }\mathscr{\Tilde{D}}_{\text{Test}}}, 
\end{align}
where $\mathscr{\Tilde{D}}_{\text{Test}}$ is got by inserting a trigger into samples in $\mathscr{D}_{test}$. 

The above quantity is sometimes termed an \emph{attack success rate (ASR)} in the machine learning literature~\cite{gao2020backdoor}. 
However, we use the notation $ASR$ to avoid ambiguity with the objective of the adversary to evade detection. 

\begin{table*}[!h]
\caption{
    Comparison between MM Trojan and baseline Trojan against the state-of-the-art MNTD detection for $t=20$ iterations for four datasets- MNIST, CIFAR-10, CIFAR-100, SpeechCommand. M represents the modification attack and B represents the blending attack. The $AUC_t$ represents the detection performance at step $t$ (smaller value indicates that the attack better evades the detection). The baseline Trojan does not have values on $AUC_{t-1}$ because it does not include an iterative process. We can observe that the MM Trojan can successfully evade the detection when the adaptive adversary takes the last step ($AUC_0$ and $AUC_{t-1}$); even when the defender takes the last step, we can still see a significant drop in the detection performance ($AUC_t$).
    }
    \label{tab:atk-mntd}
    \centering
    \begin{tabular}{c|c|c|c|c|c|c|c}
        \toprule
        Dataset & Trojan & Attack & $Acc$ & $ASR$ & $AUC_0$ ($\downarrow$) & $AUC_{t-1}$ ($\downarrow$) & $AUC_t$ ($\downarrow$) \\
        \midrule
        \multirow{4}{*}{MNIST} & \multirow{2}{*}{M} & Baseline Trojan & 0.9835 & 0.9982 & 0.9980 & - & 0.9980 \\
        \cmidrule{3-8}
         &  & MM Trojan & 0.9825 & 0.9969 & \bf 0.0 & 0.0 & \bf 0.6230 \\
        \cmidrule{2-8}
        & \multirow{2}{*}{B} & Baseline Trojan & 0.9813 & 0.9928 & 1.0 & - & 1.0 \\
        \cmidrule{3-8}
         &  & MM Trojan & 0.9813 & 0.9957 & \bf 0.0 & 0.0 & \bf 0.8310 \\
        \midrule
        \multirow{4}{*}{CIFAR-10} & \multirow{2}{*}{M} & Baseline Trojan & 0.6127 & 0.9998 & 0.9775 & - & 0.9775 \\
        \cmidrule{3-8}
         &  & MM Trojan & 0.6005 & 0.9998 & \bf 0.0 & 0.0 & \bf 0.8828 \\
        \cmidrule{2-8}
        & \multirow{2}{*}{B} & Baseline Trojan & 0.5929 & 0.9600 & 0.9658 & - & 0.9658 \\
        \cmidrule{3-8}
         &  & MM Trojan & 0.5897 & 0.8864 & \bf 0.0 & 0.0 & \bf 0.8779 \\
        \midrule
        \multirow{4}{*}{CIFAR-100} & \multirow{2}{*}{M} & Baseline Trojan & 0.4827 & 0.9988 & 0.9863 & - & 0.9863 \\
        \cmidrule{3-8}
         &  & MM Trojan & 0.4468 & 0.9409 & \bf 0.1563 & 0.0 & \bf 0.7813 \\
        \cmidrule{2-8}
        & \multirow{2}{*}{B} & Baseline Trojan & 0.4839 & 0.9992 & 0.9375 & - & 0.9375 \\
        \cmidrule{3-8}
         &  & MM Trojan & 0.4312 & 0.9108 & \bf 0.0 & 0.0 & \bf 0.8662 \\
        \midrule
        \multirow{4}{*}{SpeechCommand} & \multirow{2}{*}{M} & Baseline Trojan & 0.8330 & 0.9827 & 0.9941 & - & 0.9941 \\
        \cmidrule{3-8}
         &  & MM Trojan & 0.7481 & 0.9520 & \bf 0.1904 & 0.0361 & \bf 0.2765 \\
        \cmidrule{2-8}
        & \multirow{2}{*}{B} & Baseline Trojan & 0.8313 & 0.9911 & 0.9990 & - & 0.9990 \\
        \cmidrule{3-8}
         &  & MM Trojan & 0.7411 & 0.9703 & \bf 0.0215 & 0.0 & \bf 0.7529 \\
        \bottomrule
    \end{tabular}
\end{table*}

\underline{\emph{Detection rates ($AUC_0, AUC_{t-1}, AUC_t$):}}
The Area Under the Curve (AUC) 
in the context of a Receiver Operating Characteristic (ROC) curve is a measure of the ability of a classifier to distinguish between classes. The ROC curve is a graphical representation of true positive rate (TPR) against false positive rate (FPR) at various threshold settings. An AUC of 1.0 represents a perfect detector
An AUC of 0.5 represents a detector that performs no better than random chance.

Given the benign models $\{f_{\theta_T^i}\}_{i=1}^K$ and Trojan models $\{f_{\theta_T^i}\}_{i=1}^K$ generated in the $i$-th iteration, we will evaluate the detection AUC of different detector models on this binary classification/ outlier detection task. 
In particular, we will evaluate the AUC of (1) the vanilla detector model, denoted $AUC_0$; (2) the detectot model trained in the $(t-1)$-th iteration, denoted $AUC_{t-1}$; (3) the detector model trained in the $t$-th iteration, denoted $AUC_{t}$. 
We use $t=20$ in all our experiments. 

The quantity $AUC_{t-1}$ indicates attack performance when the adversary takes the last step, while $AUC_{t}$ indicates attack performance when the defender takes the last step. A smaller $AUC$ value indicates better attack performance. For the Baseline Trojan without multiple iterations, 
$AUC_t=AUC_0$ and there is no value for $AUC_{t-1}$.

\subsubsection{Trojan detection mechanisms}\label{subsec:MNTD}

Trojan detection mechanisms aim to identify if outputs corresponding to a set of random inputs are coming from a clean or Trojan model. 
The Trojan model detector in this paper is a binary classifier learned using multiple outputs from clean and Trojan models. 
For the purpose of experiments, SOTA Trojan detection schemes we evaluate can be broadly categorized into two classes, described below. 

The first class of detectors iteratively trains a set of clean models with same structure as target model and a set of Trojan models for different triggers and desired outputs. 
The Trojan detector then learns a discriminator on the output of these models for random image inputs. 
Such discriminator was shown to be able to achieve high accuracy in detecting Trojan models. 
MNTD~\cite{xu2021detecting} is the SOTA Trojan detector of this type.  

The second class of detectors use a `single-step' approach to determine whether a given model is Trojan or not. 
We evaluate three SOTA Trojan detectors of this type: NeuralCleanse~\cite{wang2019neural}, STRIP~\cite{gao2019strip}, and TABOR~\cite{guo2019tabor}. 
%
NeuralCleanse~\cite{wang2019neural} seeks to compute the minimum magnitude of perturbation required to change the output prediction of a model into each class. 
The given model is recognized as Trojan if the minimum magnitude of perturbation for some class is lower than a threshold value. 
STRIP~\cite{gao2019strip} aims to evaluate whether the confidence of the model on certain inputs will be maintained when these inputs will be made noisy by mixing with other inputs. If prediction of the model remain unchanged, then STRIP observes that it is likely that the input sample has been embedded with a Trojan trigger. 
Consequently, STRIP concludes that the model is Trojan. 
TABOR~\cite{guo2019tabor} is an extension of NeuralCleanse~\cite{li2021neural} and uses optimization regularization inspired by explainable AI techniques and heuristics. TABOR then uses a quality measure to identify candidate triggers, followed by an anomaly detection procedure to reduce false alarms. 

\begin{table*}[!h]
    \caption{
    Comparison between MM Trojan and baseline Trojan attacks against four SOTA Trojan detection mechanisms (MNTD~\cite{xu2021detecting}, NeuralCleanse~\cite{wang2019neural}, STRIP~\cite{gao2019strip}, and TABOR~\cite{guo2019tabor}). M represents modification attack and B represents blending attack. We use all-to-all attack in MM Trojan, as the adversary takes the last step to bypass the different detection approaches. We observe that in most settings, the MM Trojan successfully evades all four detection methods, while the baseline Trojan will fail.
    }
    \centering
    \begin{tabular}{c|c|c|c|c|c|c}
        \toprule
        Dataset & Trojan & Attack & MNTD-AUC ($\downarrow$) & NC-AUC ($\downarrow$) & STRIP-AUC ($\downarrow$) & TABOR-AUC ($\downarrow$)\\
        \midrule
        \multirow{4}{*}{MNIST} & \multirow{2}{*}{M} & Baseline Trojan & 0.9980 & 0.8750 & 0.8379 & 0.8305\\
        \cmidrule{3-7}
         &  & MM Trojan & \bf 0.0 & \bf 0.5586 & \bf 0.3674 & \bf 0.5703\\
        \cmidrule{2-7}
        & \multirow{2}{*}{B} & Baseline Trojan & 1.0 & 0.8359 & 0.6406 & 0.7188 \\
        \cmidrule{3-7}
         &  & MM Trojan  & \bf 0.0 & \bf 0.5469 & \bf 0.4205 & \bf 0.5508\\
        \midrule
        \multirow{4}{*}{CIFAR-10} & \multirow{2}{*}{M} & Baseline Trojan  & 0.9775 & 0.5859 & 0.8332 & 0.7344\\
        \cmidrule{3-7}
         &  & MM Trojan  & \bf 0.0 & \bf 0.5016 & \bf 0.5828 & \bf 0.5562\\
        \cmidrule{2-7}
        & \multirow{2}{*}{B} & Baseline Trojan  & 0.9658 & 0.6289 & 0.7141 & 0.6719\\
        \cmidrule{3-7}
         &  & MM Trojan & \bf 0.0 & \bf 0.5172 & \bf 0.4225 & \bf 0.6344\\
        \midrule
        \multirow{4}{*}{CIFAR-100} & \multirow{2}{*}{M} & Baseline Trojan & 0.9863 & 0.5754 & 0.7422 & 0.6783 \\
        \cmidrule{3-7}
         &  & MM Trojan& \bf 0.0625 & \bf 0.5188 & \bf 0.4055 & \bf 0.5413 \\
        \cmidrule{2-7}
        & \multirow{2}{*}{B} & Baseline Trojan & 0.9375 & 0.5713 & 0.7031 & 0.7154 \\
        \cmidrule{3-7}
         &  & MM Trojan & \bf 0.0 & \bf 0.5031 & \bf 0.5920 & \bf 0.5235 \\
        \midrule
        \multirow{4}{*}{SpCom} & \multirow{2}{*}{M} & Baseline Trojan & 0.9941 & 0.7070 & 0.7514 & 0.7273\\
        \cmidrule{3-7}
         &  & MM Trojan & \bf 0.0156 & \bf 0.6055 & \bf 0.2187 & \bf 0.6600 \\
        \cmidrule{2-7}
        & \multirow{2}{*}{B} & Baseline Trojan & 0.9990 & 0.9297 & 0.9188 & 0.9143\\
        \cmidrule{3-7}
         &  & MM Trojan & \bf 0.0 & \bf 0.6680 & \bf 0.3748 & \bf0.6667\\
        \bottomrule
    \end{tabular}
    \label{tab:atk-other}
\end{table*}

\subsection{MM Trojan vs Iterative Detectors}\label{subsec:EvalMNTD}

Table~\ref{tab:atk-mntd} shows the attack performance of \emph{MM Trojan} and a comparison with the baseline Trojan on the MNIST, CIFAR-10, CIFAR-100, and SpeechCommand datasets. We observe that MM Trojan successfully evades adaptive Trojan detection models. 
We use MNTD as an instance of iterative Trojan detection mechanism \cite{xu2021detecting}. 
We report results for both, modification (M) and blending (B) attacks. The model trained with MM Trojan can evade the initial MNTD detector into a zero detection AUC on vision-based tasks, and close-to-zero performance on the SpeechCommand task (values of $AUC_0$ in Table~\ref{tab:atk-mntd}). For an adaptive adversary, using MM Trojan results in a zero or close-to-zero detection AUC at step $t-1$ (values of $AUC_{t-1}$ in Table~\ref{tab:atk-mntd}). 
Moreover, when the defender takes the last step, we can still see a significant drop in detection AUC ($AUC_t$ values in Table~\ref{tab:atk-mntd}). This reveals that even when the defender has full knowledge of our attack, the defending performance is still significantly affected by the attack. 

We also observe a trade-off in that our MM Trojan will, in some cases, have lower benign and backdoor accuracies, especially on relatively complicated tasks such (e.g., SpeechCommand). We believe that tuning parameters $\theta_D$ and $\theta_T$ will control the trade-off. Developing strategies to mitigate drop in accuracy is an interesting direction of future research. 

\subsection{MM Trojan vs Other Detectors}\label{subsec:EvalOther}

We additionally evaluate the effectiveness of an MM Trojan model in evading other Trojan detection approaches. Since MM Trojan does not change the training dataset but only the model training algorithm, it only affects model-based Trojan detection algorithms. Therefore, we perform experiments on three model-based detection algorithms: NeuralCleanse~\cite{wang2019neural}, STRIP~\cite{gao2019strip}, and TABOR~\cite{guo2019tabor}. In order to evade these detection methods, we will change the attack goal of the Trojan to be the ``\emph{all-to-all}'' attack, i.e., the prediction will be changed from the $i$-th class to the $((i+1)\%K)$-th class, where $K$ is the total number of classes. 
Our approach is guided by an insight observed in~\cite{xu2021detecting} which has indicated that 
model-based detection algorithms are not designed to detect all-to-all attacks. 
Since MM Trojans are designed for an adaptive adversary that uses gradient-based algorithms, we believe that it will also be effective against NeuralCleanse or STRIP. We seek to show that MM Trojan can be flexibly combined with other techniques, so that it can also bypass other detection models.
Table~\ref{tab:atk-other} shows results of the MM \emph{all-to-all} Trojan. We observe that with the all-to-all attack goal integrated into MM Trojan, we can indeed achieve good performance in evading multiple detectors. 
The AUC values of NeuralCleanse, STRIP, and TABOR drop to under or close to $0.5$ (equivalent to a random guess) in all cases. 
In comparison, the AUC value of MNTD is close to zero because of the adversarial attack in the MM Trojan process. 
We acknowledge that as our approach does not involve gradient-based attacks against NeuralCleanse, STRIP, or TABOR, we anticipate that these models will maintain a performance above close-to-zero detection rates.

The results in this section on three SOTA Trojan detectors indicate that success of an adversary following MM Trojan is largely agnostic to the detection mechanism. This suggests that MM Trojan can be effective against future Trojan detection mechanisms that adaptively update their parameters when working against an adaptive adversary. 

\begin{figure*}[!h] 
\begin{subfigure}{0.24\textwidth}
\includegraphics[width=\linewidth]{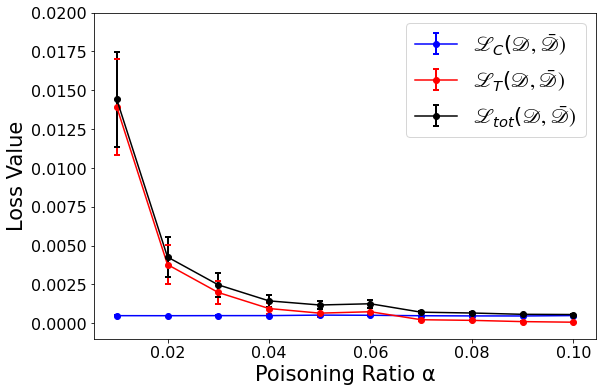}
\caption{MNIST Trojan-M} \label{fig:a}
\end{subfigure}
\begin{subfigure}{0.24\textwidth}
\includegraphics[width=\linewidth]{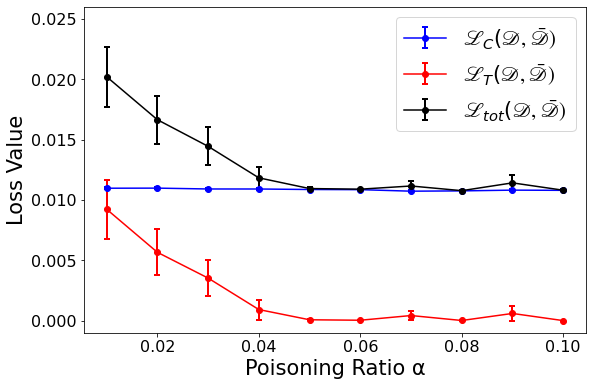}
\caption{CIFAR-10 Trojan-M} \label{fig:a}
\end{subfigure}
\begin{subfigure}{0.24\textwidth}
\includegraphics[width=\linewidth]{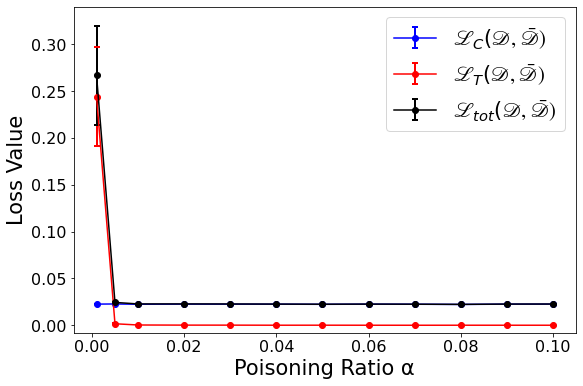}
\caption{CIFAR-100 Trojan-M} \label{fig:b}
\end{subfigure}
\begin{subfigure}{0.24\textwidth}
\includegraphics[width=\linewidth]{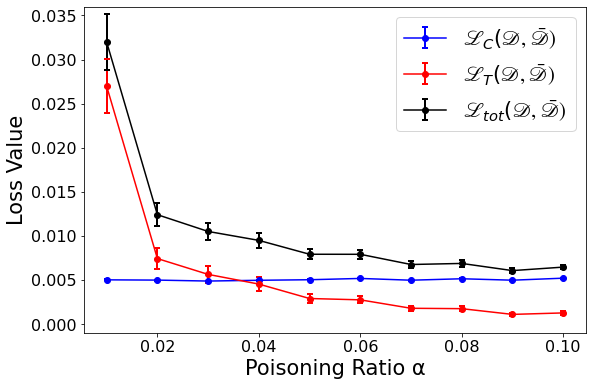}
\caption{SpeechCommand Trojan-M} \label{fig:b}
\end{subfigure}

\medskip
\begin{subfigure}{0.24\textwidth}
\includegraphics[width=\linewidth]{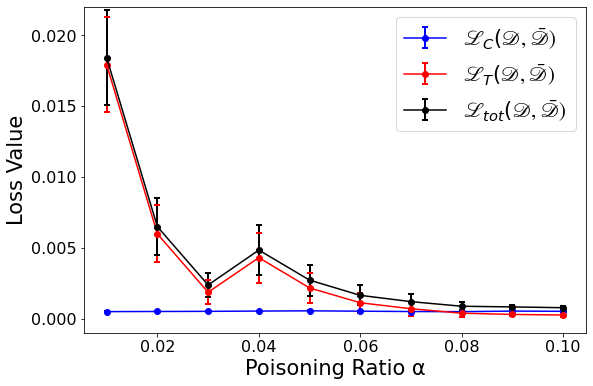}
\caption{MNIST Trojan-B} \label{fig:a}
\end{subfigure}
\begin{subfigure}{0.24\textwidth}
\includegraphics[width=\linewidth]{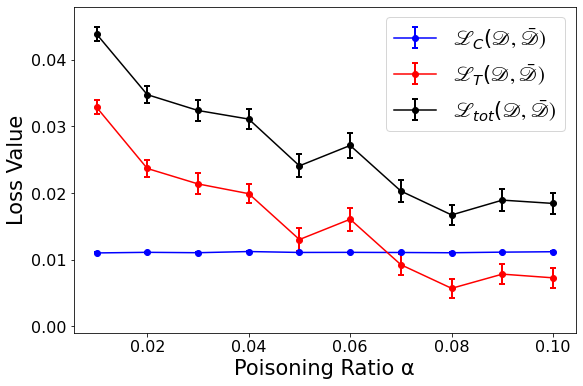}
\caption{CIFAR-10 Trojan-B} \label{fig:c}
\end{subfigure}
\begin{subfigure}{0.24\textwidth}
\includegraphics[width=\linewidth]{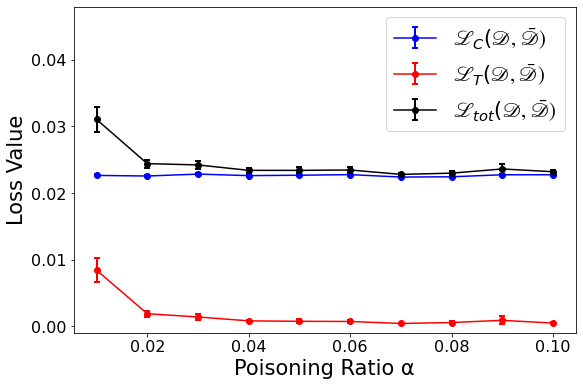}
\caption{CIFAR-100 Trojan-B} \label{fig:d}
\end{subfigure}
\begin{subfigure}{0.24\textwidth}
\includegraphics[width=\linewidth]{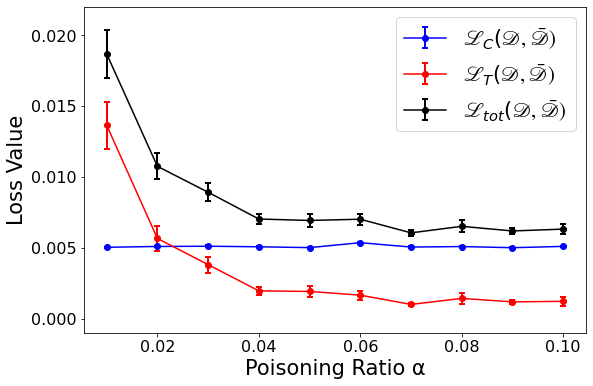}
\caption{SpeechCommand Trojan-B} \label{fig:b}
\end{subfigure}
\caption{
Adversary test-loss as a function of the fraction of data samples into which a Trojan trigger has been embedded on MNIST, CIFAR-10, CIFAR-100, and SpeechCommand datasets under modification (M) and blending attacks (B). We denote the fraction of samples into which the adversary embeds a Trojan trigger by $\alpha:=|\Tilde{\mathscr{D}}|/|\mathscr{D}|$. \textcolor{blue}{Blue} curves show the loss term over clean samples, $\mathcal{L}_{C}(\mathscr{D},\mathscr{\Tilde{D}})$, and \textcolor{red}{red} curves show the loss term over samples that have been embedded with the Trojan, $\mathcal{L}_{T}(\mathscr{D},\mathscr{\Tilde{D}})$. The black curves present the total loss $\mathcal{L}_{tot}(\mathscr{D},\mathscr{\Tilde{D}})$ as defined in Eqn. \eqref{eq:Total_loss}. We observe that for both attacks, the decrease in the total loss value, $\mathcal{L}_{tot}(\mathscr{D},\mathscr{\Tilde{D}})$,  is insignificant after exceeding a threshold value of $\alpha$. This indicates that the adversary does not derive an additional (marginal) benefit by increasing the fraction of Trojan trigger-embedded samples beyond a threshold value of $\alpha$.  
}
\label{fig:subloss}
\end{figure*}

\subsection{Verification of Greedy Trojan}\label{subsec:EvalSubmod}

\begin{figure*}[!h] 
\begin{subfigure}{0.24\textwidth}
\includegraphics[width=\linewidth]{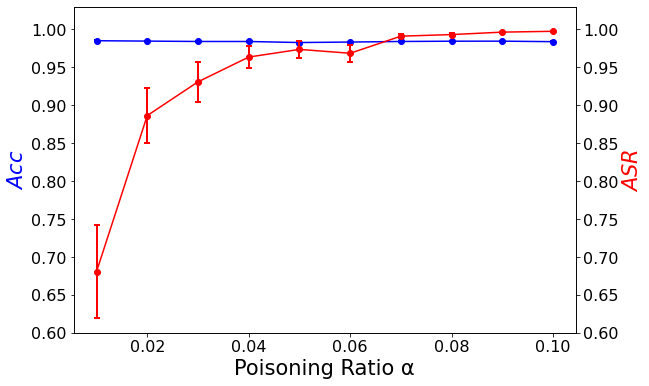}
\caption{MNIST Trojan-M} \label{fig:a}
\end{subfigure}
\begin{subfigure}{0.24\textwidth}
\includegraphics[width=\linewidth]{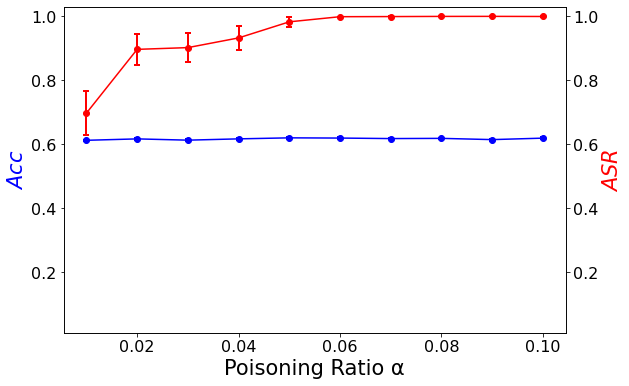}
\caption{CIFAR-10 Trojan-M} \label{fig:a}
\end{subfigure}
\begin{subfigure}{0.24\textwidth}
\includegraphics[width=\linewidth]{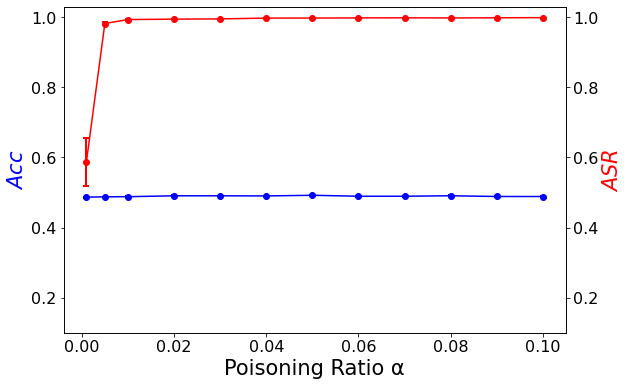}
\caption{CIFAR-100 Trojan-M} \label{fig:b}
\end{subfigure}
\begin{subfigure}{0.24\textwidth}
\includegraphics[width=\linewidth]{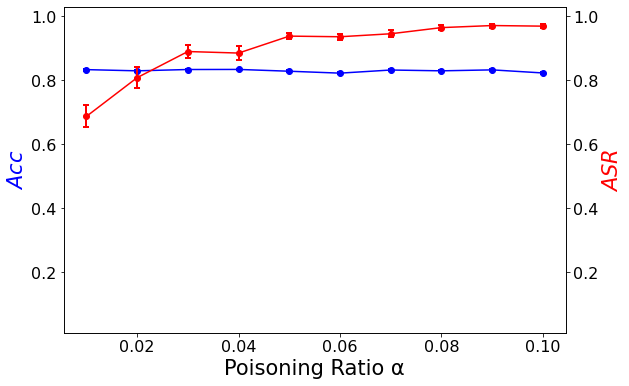}
\caption{SpeechCommand Trojan-M} \label{fig:b}
\end{subfigure}

\medskip
\begin{subfigure}{0.24\textwidth}
\includegraphics[width=\linewidth]{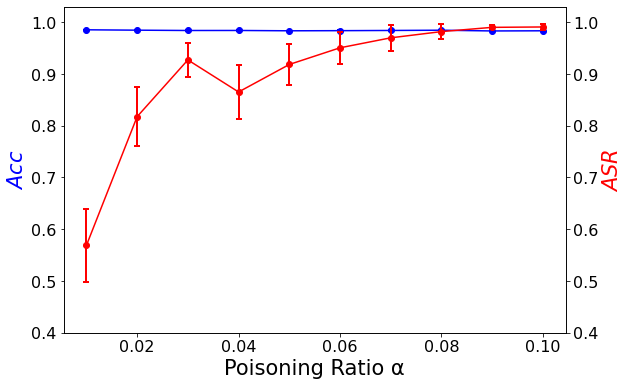}
\caption{MNIST Trojan-B} \label{fig:a}
\end{subfigure}
\begin{subfigure}{0.24\textwidth}
\includegraphics[width=\linewidth]{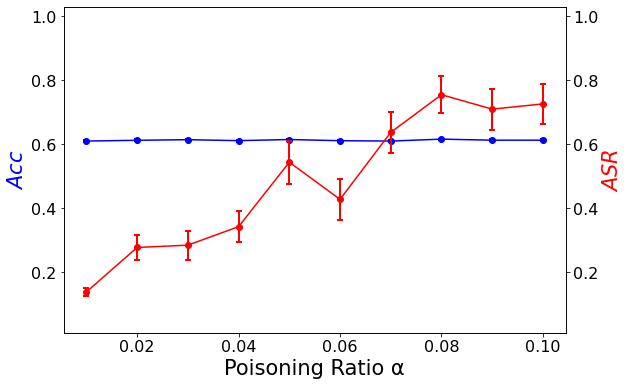}
\caption{CIFAR-10 Trojan-B} \label{fig:c}
\end{subfigure}
\begin{subfigure}{0.24\textwidth}
\includegraphics[width=\linewidth]{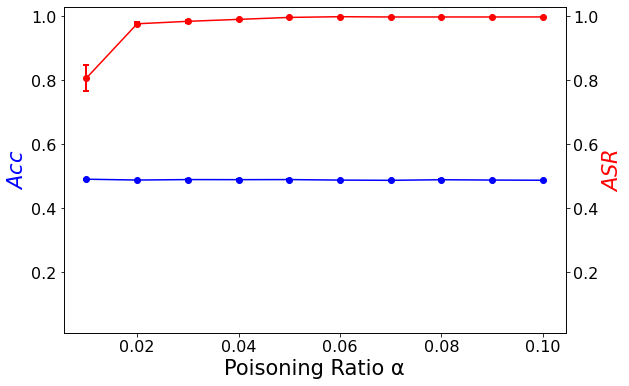}
\caption{CIFAR-100 Trojan-B} \label{fig:d}
\end{subfigure}
\begin{subfigure}{0.24\textwidth}
\includegraphics[width=\linewidth]{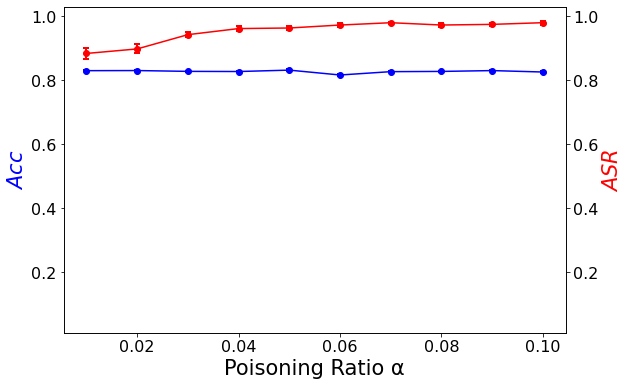}
\caption{SpeechCommand Trojan-B} \label{fig:b}
\end{subfigure}
\caption{Classification accuracies over clean samples ($Acc$, \textcolor{blue}{blue}) and samples into which the adversary embeds a Trojan trigger ($ASR$, \textcolor{red}{red})  using the MM Trojan algorithm on MNIST, CIFAR-10, CIFAR-100, and SpeechCommand datasets under modification (M) and blending (B) attacks. The increase in $ASR$ becomes less significant beyond a threshold value of the poisoning ratio, $\alpha:=|\Tilde{\mathscr{D}}|/|\mathscr{D}|$. On the other hand, the value of $Acc$ may decrease as $\alpha$ increases. This could result in a user discarding the model. 
}
\label{fig:subacc}
\end{figure*}

\begin{table*}[!h]
    \caption{
    Ablation study of MM Trojan with different poisoning ratio ($\alpha$) against the state-of-the-art MNTD detection. M represents the modification attack and B represents the blending attack.  We can observe that under all the settings, a minimum threshold of $\alpha=0.05$ can achieve a good Trojan attack performance ($Acc$ and $ASR$) while also evading MNTD detection.
    }
    \label{tab:abb}
    \centering
    \begin{tabular}{c|c|c|c|c|c|c}
        \toprule
        Dataset & Trojan & Metric & $\alpha = 0.05$ & $\alpha = 0.1$ & $\alpha = 0.25$ & $\alpha = 0.5$ \\
        \midrule
        \multirow{4}{*}{MNIST} & \multirow{2}{*}{M} & $Acc$ & 0.9828 & 0.9821  & 0.9819 & 0.9812   \\
        \cmidrule{3-7}
         &  & $ASR$ & 0.9746 & 0.9951  &  0.9977  & 0.9995   \\
        \cmidrule{3-7}
         &  & AUC ($\downarrow$) &  \bf 0.0 &  \bf 0.0 &  \bf 0.0 &  \bf 0.0   \\
        \cmidrule{2-7}
        & \multirow{2}{*}{B} & $Acc$ & 0.9826   & 0.9819  & 0.9813 & 0.9798  \\
        \cmidrule{3-7}
         &  & $ASR$ & 0.9580 & 0.9943 & 0.9970  & 0.9975   \\
         \cmidrule{3-7}
         &  & AUC ($\downarrow$) &  \bf 0.0 &  \bf 0.0 & \bf 0.0 &  \bf 0.0   \\
        \midrule
        \multirow{4}{*}{CIFAR-10} & \multirow{2}{*}{M} & $Acc$ & 0.6015 & 0.5888   & 0.5870  & 0.5818   \\
        \cmidrule{3-7}
         &  & $ASR$ & 0.9987  & 0.9998  & 0.9999  & 1.0    \\
         \cmidrule{3-7}
         &  & AUC ($\downarrow$) &  \bf 0.0 &  \bf 0.0 & \bf 0.0 &  \bf 0.0  \\
        \cmidrule{2-7}
        & \multirow{2}{*}{B} & $Acc$ & 0.6150  & 0.5948  &  0.5905 & 0.5875  \\
        \cmidrule{3-7}
         &  & $ASR$ & 0.6186  & 0.8700  & 0.9368  & 0.9994    \\
         \cmidrule{3-7}
         &  & AUC ($\downarrow$) &  \bf 0.0 &  \bf 0.0 & \bf  \bf 0.0 &  \bf 0.0   \\
        \midrule
        \multirow{4}{*}{CIFAR-100} & \multirow{2}{*}{M} & $Acc$ & 0.4887   & 0.4881  & 0.4867  & 0.4838    \\
        \cmidrule{3-7}
         &  & $ASR$ & 0.9931 & 0.9987   & 0.9998    & 1.0   \\
         \cmidrule{3-7}
         &  & AUC ($\downarrow$) & \bf 0.0089 & 0.0356   & 0.0581  &  0.0982  \\
        \cmidrule{2-7}
        & \multirow{2}{*}{B} & $Acc$ & 0.4892  & 0.0.4871  & 0.4838  & 0.4798   \\
        \cmidrule{3-7}
         &  & $ASR$ & 0.8065  & 0.9978  & 0.9993  & 1.0  \\
         \cmidrule{3-7}
         &  & AUC ($\downarrow$) & \bf 0  & \bf 0  & 0.0001  & 0.0024   \\
        \midrule
        \multirow{4}{*}{SpeechCommand} & \multirow{2}{*}{M} & $Acc$ & 0.7638  & 0.7479   & 0.7069  &  0.6982  \\
        \cmidrule{3-7}
         &  & $ASR$ & 0.7839  & 0.9270 & 0.9443 &  0.9951  \\
         \cmidrule{3-7}
         &  & AUC ($\downarrow$) & \bf 0.0127 & 0.0264  & 0.0286   & 0.0345   \\
        \cmidrule{2-7}
        & \multirow{2}{*}{B} & $Acc$ & 0.7596  & 0.7591  & 0.7256  & 0.7091   \\
        \cmidrule{3-7}
         &  & $ASR$ & 0.9382  & 0.9427  & 0.9554  & 0.9585    \\
         \cmidrule{3-7}
         &  & AUC ($\downarrow$) & \bf 0.0022 & 0.0028  & 0.0415 &  0.0600   \\
        \bottomrule
    \end{tabular}
\end{table*}

Fig.~\ref{fig:subloss} and \ref{fig:subacc} present experimental evaluations of the Greedy Trojan Algorithm on the  MNIST, CIFAR-10, CIFAR-100, and SpeechCommand datasets for two different types of attacks- a modification attack (M) and a blending attack (B). 
We define the poisoning ratio $\alpha$ to be the fraction of samples into which the adversary embeds a Trojan trigger, i.e., $\alpha:=|\Tilde{\mathscr{D}}|/|\mathscr{D}|$. 
In all experiments, we initialize $\alpha=0.002$. 
In each iteration, we train a model to minimize the loss function in Eqn. (\ref{eq:Total_loss}). 
This loss function consists of two terms- one over clean samples and another over Trojan samples. 
We increase the size of the set of samples, $\Tilde{\mathscr{D}}$, into which a Trojan trigger is embedded by the adversary as long as the addition of such samples decreases the value of the loss function (Line 4 of Algo \ref{alg:greedy}). 
The procedure terminates when adding more samples to $\Tilde{\mathscr{D}}$ does not result in a significant decrease in the loss function. 

Fig.~\ref{fig:subloss} shows the change in loss values as the value of $\alpha$ increases. 
 In each case, we observe that the loss term corresponding to Trojan samples ($\mathcal{L}_{T}(\mathscr{D},\mathscr{\Tilde{D}})$, red curve) initially decreases as $\alpha$ increases. 
 However, after exceeding a threshold, the decrease in the total loss value ($\mathcal{L}_{tot}(\mathscr{D},\mathscr{\Tilde{D}})$, black curve) is insignificant. 
 The loss term corresponding to clean samples ($\mathcal{L}_{C}(\mathscr{D},\mathscr{\Tilde{D}})$, blue curve) might increase as fraction of clean samples decreases with increase in $\alpha$. 
 This demonstrates that the adversary has little advantage in poisoning a fraction of samples larger than a threshold value. 
 
Fig.~\ref{fig:subacc} shows classification accuracies for clean samples ($Acc$, blue curve) and Trojan samples ($ASR$, red curve) as $\alpha$ is increased.
 For small values of $\alpha$, $ASR$ is small, but this metric improves as $\alpha$ is increased. 
 The improvement becomes less significant after a threshold of $\alpha$ is reached. 
 Simultaneously, the value of $Acc$ decreases slightly, which indicates that the model might become less useful to a user, leading them to discard it. Thus poisoning a large fraction of samples will not be advantageous to the adversary. 
 
 The outcomes of our experiments is consistent with the analysis presented in Section \ref{sec:Submod} in that the test-loss function of the adversary exhibits diminishing returns as the number of samples chosen by the adversary to be embedded with the Trojan trigger increases. 
 Further, these results show that the fraction of samples into which the adversary needs to embed a Trojan trigger can be determined in a constructive manner. 
 We observe that the adversary needs to embed a Trojan into only about $4\%$ of samples under the modification attack and about $6\%$ of samples under the blending attack. 
 Increasing the poisoning ratio beyond these thresholds will not significantly reduce the value of the total loss, $\mathcal{L}_{tot}(\mathscr{D},\mathscr{\Tilde{D}})$, in Eqn. (\ref{eq:Total_loss}). 

%

\subsection{Ablation Studies}

We conduct ablations to evaluate the effect of poisoning ratio $\alpha$ on performance of the MM Trojan algorithm using MNIST, CIFAR-10, CIFAR-100, and SpeechCommand datasets. 
For an adversary carrying out two different types of attacks to embed a Trojan trigger- modification (M) and blending attack (B)- we consider four values: $\alpha=\{0.05,0.1,0.25,0.5\}$, and examine benign accuracy $(Acc)$, backdoor accuracy $(ASR)$, and detection rate of MNTD ($AUC$). 

Table~\ref{tab:abb} shows that for both types of attacks, as the poisoning ratio $\alpha$ is increased, the value of $Acc$ decreases while the value of $ASR$ increases. 
For the MNIST and CIFAR-10 datasets, the $AUC$ values of MNTD are zero, indicating that the adversary successfully evades detection even when the fraction of poisoned samples is high. 
We also observe that the adversary needs to maintain a poisoning ratio of at least $0.1$ in order to achieve high backdoor accuracy values for these two datasets. 
On the other hand, for the SpeechCommand dataset, there is a tradeoff between achieving high backdoor accuracy $ASR$ and low $AUC$ values, indicating that the poisoning ratio has a greater impact. 
However, for all four datasets, we observe that there is a threshold value of $\alpha$, beyond which there is a only a marginal change in the values of $Acc$ and $ASR$. 

The results in Table~\ref{tab:abb} further underscore our insight that a minimum threshold for $\alpha$ can achieve acceptable values of $Acc$ and $ASR$ using the Greedy Algorithm (Algo.~\ref{alg:greedy}) and also aids the adversary in effectively evading detection using the MM Trojan algorithm (Algo. \ref{alg:minmax}). 

\section{Discussion}\label{sec:Discussion}
In this section, we provide insights that underscore our results 
and briefly describe open questions that are promising directions for future research. 

\noindent 
\underline{\textbf{Performance of MM Trojan:}} 
Our MM Trojan algorithm (Algorithm~\ref{alg:minmax}) jointly optimizes parameters associated with the Trojan model ($\theta_T$) and detection mechanism ($\theta_D$) in order for an adaptive adversary to simultaneously evade detection and achieve high classification accuracy. The adaptive defender aims to maximize the probability of detecting Trojan models based on their outputs. We model this interaction as an iterated game and in Proposition~\ref{prop:NE1} we show solving this game results in the adversary successfully evading detection. Specifically, the solution $\theta_T^{\*}$ and $\theta_D^{\*}$ to the game forms a Nash equilibrium under which neither adversary nor defender can do better by unilaterally deviating from their respective solutions. Therefore, as long as the adversary uses the value of $\theta_T^{\*}$ output by the MM Trojan algorithm, any other defender (including a `static' defender) will not be able to deviate from $\theta_D^{\*}$ and achieve better performance. 

\noindent 
\underline{\textbf{Training overhead:}} 
The number of hours of training required to generate shadow models and train meta models for $t=20$ iterations are: $MNIST: 67.2$; $CIFAR: 78.1$; $SpeechCommand: 46.4$. Batch processing is one way to reduce the training overhead, as noted in Remark \ref{RemBatchProc}.

\noindent 
\underline{\textbf{Information exclusively controlled by adversary:}}  
Fig. \ref{fig:GAN-Train-AI-Trojan} indicates that in order to train Trojaned models, the adversary uses (i) data samples embedded with different types of Trojan triggers and (ii) feedback from the output of the detector. 
On the other hand, training the detector requires outputs of Trojaned models and outputs from clean models. 
Thus, the adversary has control over a subset of information (corresponding to outputs of the Trojaned models) required to train the detector. 
Further, 
to avoid imbalanced data, an equal number of outputs of data samples from clean  models and of samples from Trojaned models are used to train the detector. 
The detector on the other hand, makes use of output data from clean and Trojaned DNNs, and does not generate any additional data of it own.

These observations are consistent with our
experimental results in Sec. \ref{subsec:EvalMNTD} and \ref{subsec:EvalOther} which show that an adaptive adversary following the MM Trojan Algorithm is able to successfully evade detection by multiple different SOTA Trojan detectors. 

\section{Related Work}
\label{sec:relatedwork}
We place our contributions in this paper in the context of related work in backdoor attacks on ML models, detecting Trojan models, and mitigating the impact of backdoors.  

\noindent 
\textbf{\underline{Trojan Attacks.}}
The authors of~\cite{ji2018transferlearningtrojan} showed that untrustworthy entities might maintain pre-trained ML models, consequently resulting in severe security implications to model users. 
They introduced a class of \emph{model-reuse} attacks, wherein a host system running this model could be induced to predictably misbehave on specific inputs. 
A notion of \emph{latent backdoors} was introduced in~\cite{yao2019latent}, which are backdoors that are preserved during transfer learning. 
As a result, if a model obtained after transfer learning includes the target label of the backdoor, it can be activated by an appropriate input. 
This work also showed that defending against latent backdoors involved a tradeoff between cost and accuracy. 

Model or data poisoning by an adversary has also been shown to be an effective attack strategy. 
An algorithm in~\cite{rakin2020tbt} efficiently identified vulnerable bits of model weights stored in memory. 
Flipping these bits transformed a deployed DNN model into a Trojan model. 
Concurrently, targeted weight perturbations were used to embed backdoors into convolutional neural networks deployed for face recognition in~\cite{dumford2020backdooring}. 
In contrast to model poisoning methods described above, the training set of an ML model was poisoned in~\cite{shafahi2018poison}. 
In this case, the adversary had the ability to control the outcome of classification for the poisoned data, while the model's performance was maintained for `clean' samples. 

In \cite{salem2022dynamic}, the authors propose a backdoor trigger-embedding approach where an intended trigger is first partitioned into smaller triggers. 
Each smaller trigger is embedded into a unique input sample such that the union of the trigger-embedded data covers the intended trigger. 
The results in \cite{salem2022dynamic} also show that the resulting Trojaned DNN leads to effective mixing of trigger-embedded and clean input samples in the classification space making it difficult for any SOTA distance-based detection mechanisms. 
In \cite{salem2022dynamic}, it was noted that at least 30\% trigger-embedded data is required to ensure high accuracy on both clean and trigger-embedded input samples. 


\noindent 
\textbf{\underline{Detecting Trigger-Embedded Inputs.}} 
An intuitive defense against data poisoning attacks is to develop techniques to remove suspicious samples from the training data. 
Backdoor attacks were shown to leave a distinct \emph{spectral signature} in covariance of feature representations learned by DNNs  in~\cite{tran2018spectral}. 
The spectral signature, in addition to detecting corrupted training examples, also presents a barrier to crafting and design of backdoors to ML models. 
Gradients of loss function at the input layer of a DNN were used to extract signatures to identify and eliminate poisoned data in~\cite{chan2019poison} even when the target class and ratio of poisoned samples were unknown. 
{
Recent work in \cite{liu2023detecting} and \cite{guo2023scale} employ corruption robustness consistency and prediction consistency-based measures to detect trigger-embedded input samples during inference. Latent separation based methods for detecting Trojan-trigger embedded input samples, aiming to learn separable latent representations for Trojaned and clean inputs were proposed in \cite{tang2021demon, hayase2021spectre}. However these approaches which focus on input-based or latent separability (i.e., model inspection)-based detection of Trojan trigger-embedded input samples, are not within the scope of our adaptive adversary model. Our objective is to evade output-based Trojaned model detectors, which aim to identify whether a given model is Trojaned without requiring trigger-embedded Trojan input samples for detection. 
}

\noindent 
\textbf{\underline{Detecting Trojaned Models.}}
When the ML model has an embedded Trojan, it is natural to design methods to detect such models by identifying potential triggers. 
A method to detect and reverse-engineer triggers in DNNs called \emph{NeuralCleanse} was proposed in~\cite{wang2019neural}. 
\emph{NeuralCleanse} devised an optimization scheme to determine the smallest trigger required to misclassify all samples from all labels to the target label. 
However, the target label may not be known, and the requirement to consider all labels as a potential target label makes this computationally expensive. 
An improvement on \emph{NeuralCleanse} called \emph{TABOR}~\cite{guo2019tabor} used a combination of heuristics and regularization techniques to reduce false positives in Trojan detection. 
Differences in explanations of outputs of a clean model and a Trojan model, even on clean samples, were used to identify Trojan models in~\cite{huang2019neuroninspect}. 
This was achieved without requiring access to samples containing a trigger. 

A Trojan detection method called \emph{DeepInspect}~\cite{chen2019deepinspect} used conditional generative models to learn distributions of potential triggers while having only black-box access to the deployed ML model and without requiring access to clean training data. 
To overcome a limitation of~\cite{chen2019deepinspect} (and of~\cite{wang2019neural}) that there is exactly one target label, the authors of~\cite{xu2021detecting} proposed Meta-Neural Trojan Detection (MNTD). 
MNTD only required black-box access to models and trained a meta-classifier to identify if the model was Trojan or not. 
Multiple shadow models (combination of clean and Trojan models) were generated and their representations were used to learn a binary classifier. 
At test-time, the representation of the target model was provided to the classifier to determine whether it was Trojan or not. 
MNTD was demonstrated to have a high success rate when evaluated on a diverse range of datasets, including images, speech, and NLP. 

\noindent 
\textbf{\underline{Mitigating Backdoors.}}
Once an embedded Trojan has been detected, a question arises if methodologies can be developed to remove or reduce the impact of a backdoor, while maintaining model performance on clean inputs. 
The \emph{Fine-pruning} method in~\cite{liu2018fine} proposes to deactivate neurons that are not enabled by clean inputs. It then uses a tuning procedure to restore some of the deactivated neurons to mitigate reduction in classification accuracy due to pruning. 
A run-time Trojan attack detection system called STRIP was developed in~\cite{gao2019strip}. 
The \emph{input-agnostic} property of a trigger was exploited as a weakness of Trojan attacks. 
This method was independent of model architecture and size of trigger. 
A survey of backdoor attacks and defenses against such attacks can be found in~\cite{gao2020backdoor}. 



\section{Conclusion}
\label{sec:Conclusion}

In this paper, we investigated the detection performance of SOTA output-based Trojaned model detectors against an adaptive adversary who has the knowledge of the deployment of such detectors and evolves its strategies to bypass detection. 

Such an adaptive adversary incorporates detector parameter information to retrain the Trojaned DNN to (1) achieve high accuracy on both Trojan trigger-embedded and clean input samples and (2) bypass detection. 
By allowing detectors to also be adaptive, 
we showed that co-evolution of adversary and detector parameters could be modeled by an iterative game. 
We proved that the solution of this game resulted in the adversary accomplishing objectives (1) and (2). 
We also used a greedy algorithm to allow the adversary to select 
input samples to embed a Trojan trigger. 
When cross-entropy or log-likelihood loss functions were used for training, the greedy algorithm resulted in a provable lower bound on the number of samples to be selected for trigger-embedding. %

Extensive evaluations on MNIST, CIFAR-10, CIFAR-100, and SpeechCommand showed that the adaptive adversary effectively evaded four leading SOTA output-based Trojaned model detectors: MNTD~\cite{xu2021detecting}, NeuralCleanse~\cite{wang2019neural}, STRIP~\cite{gao2019strip}, and TABOR~\cite{guo2019tabor}. 

Our results highlight the need for new, advanced output-based Trojaned model detectors against adaptive adversaries. 

\section*{Acknowledgments}

This material is based upon work supported by the National Science
Foundation under grant IIS 2229876 and is supported in part by
funds provided by the National Science Foundation (NSF), by the
Department of Homeland Security, and by IBM. Any opinions,
findings, and conclusions or recommendations expressed in this
material are those of the author(s) and do not necessarily reflect
the views of the National Science Foundation or its federal agency
and industry partners. This work is also supported by the Air Force
Office of Scientific Research (AFOSR) through grant FA9550-23-1-
0208, the Office of Naval Research (ONR) through grant N00014-23-
1-2386, and the NSF through grant CNS 2153136.

\bibliographystyle{splncs04}      
\bibliography{GameSec2022_references}
\appendix
\section{Appendix}
This Appendix presents detailed proofs of our result in Proposition \ref{prop:NE1} and supermodularity of the test loss (Theorem \ref{thm:supermodular}). 

\subsection{Proof of Proposition~\ref{prop:NE1}}
Proposition \ref{prop:NE1} provided a characterization of the solution of the min-max optimization problem presented in Eqn. (\ref{eq:min-max-modified}). We show that output distributions of Trojaned and clean DNNs will be identical at the optimal solution, and can be interpreted as the adversary winning the co-evolution game by successfully evading detection. 
%
We restate Prop. \ref{prop:NE1} below. 

\noindent 
\textbf{Proposition \ref{prop:NE1}}: 
\emph{For a random input data sample $x \in \mathscr{D}_{R}$, let $z_{T} := f_{\theta_T}(x)$  and $z_{C} := f_{\theta_C}(x)$ respectively denote the outputs of a Trojaned model and a clean model. Let $q_{T}$ and $q_{C}$ denote probability distributions associated with $z_{T}$ and $z_{C}$. Then, at the optimal solution of the game in Eqn.~\eqref{eq:min-max-modified}, the output distributions coming from clean models and Trojaned models will be identical- i.e., $ q_{T} =  q_{C}$, thus allowing the adversary to successfully evade detection.}

\begin{proof}
Eqn. (\ref{eq:min-max-modified}) consists of four terms: 
\begin{align*}
    \hspace{-5mm} \min_{\theta_T} \max_{\theta_{D}} \hspace{1mm} &\mathbb{E}_{z_{T} \sim q_{T}} [\log (1-h_{\theta_D}(z_{T}))]    \\&+  \mathbb{E}_{z_{C} \sim q_{C}} [\log (h_{\theta_D}(z_{C}))]  \\&+
    \mathbb{E}_{x \sim {p}} [\ell_{\theta_T} (x \times (1-\Delta)+ \delta \times \Delta, y_T) + \ell_{\theta_T}(x, y_C)], 
\end{align*}
where $x_T:=x \times (1-\Delta)+ \delta \times \Delta$ denotes a trigger-embedded input sample. 
For a fixed $\theta_T$, let $V_{\text{max}}$ denote the solution to the 
maximization part of the min-max problem in Eqn.~\eqref{eq:min-max-modified}. We have:
\begin{eqnarray}\label{eq:Dval}
V_{\text{max}} =     &&\hspace{-5mm}
    \mathbb{E}_{z_{T} \sim {q}_{T}} [\log (1-h_{\theta_D}(z_T))] + \mathbb{E}_{z_{C} \sim {q}_{C}} [\log (h_{\theta_D}(z_C))] 
    \nonumber \\
    &\hspace{-9.5mm}=&\hspace{-5mm}
    \mathbb{E}_{z \sim {q}_{T}} [\log (1-h_{\theta_D}(z))] + \mathbb{E}_{z \sim {q}_{C}} [\log (h_{\theta_D}(z))] 
    \nonumber \\
    &\hspace{-9.5mm}=&\hspace{-5mm}
    \hspace{-1mm}\int_{z} {q}_T(z)  \log (1-h_{\theta_D}(z)) dz\nonumber + \hspace{-1.5mm}\int_{z} {q}_C(z) \log (h_{\theta_D}(z)) dz
    , \nonumber 
\end{eqnarray}
where the first two equalities hold by variable substitution, and the last equality holds by the definition of expectation.

To complete solving the maximization sub-problem, we observe that for any $(a, b) \in \mathbb{R}^{2} \backslash\{0,0\}$, the function $x \rightarrow a\log(1-x) + b\log(x)$ achieves its maximum in $[0,1]$ at $\frac{b}{a+b}$. 
Therefore, for a fixed $\theta_T$, we obtain 
\begin{align*}
    h_{\theta_D^{\*}}(z) = \frac{{q}_C(z)}{{q}_T(z)+ {q}_C(z)}. 
\end{align*}
With the above value of $h_{\theta_D^{\*}}(z)$, let $V_{\text{min}}$ denote the solution to the 
minimization part of the min-max optimization problem in Eqn.~\eqref{eq:min-max-modified}. We have 
\begin{align}
    V_{\text{min}}&=\int_{z} {q}_T(z)  \log \Big(1-\frac{{q}_C(z)}{{q}_T(z)+{q}_C(z)}\Big) dz \nonumber \\
    &\quad + \int_{z} {q}_C(z) \log \Big(\frac{{q}_C(z)}{{q}_T(z)+{q}_C(z)}\Big) dz \nonumber \\
    &\quad +\mathbb{E}_{x \sim {p}} [\ell_{\theta_T} (x \times (1-\Delta)+ \delta \times \Delta, y_T) + \ell_{\theta_T}(x, y_C)] \nonumber \\
    &= \int_{z} {q}_T(z)  \log \Big(\frac{{q}_T(z)}{{q}_T(z)+{q}_C(z)}\Big) dz \nonumber \\ &
    \quad + \int_{z} {q}_C(z) \log \Big(\frac{{q}_C(z)}{{q}_T(z)+{p}_C(z)}\Big) dz 
    \label{Eq:IntermediateStep} \\
    &\quad +
    \mathbb{E}_{x \sim {p}} [\ell_{\theta_T} (x \times (1-\Delta)+ \delta \times \Delta, y_T) + \ell_{\theta_T}(x, y_C)] \nonumber 
\end{align}
The first two terms in Eqn.~\eqref{Eq:IntermediateStep} can be interpreted as KL divergences between probability distributions $q_T$ and $q_C$. 
The KL divergence between two probability distributions $p_1$ and $p_2$ is defined by $KL(p_1||p_2) := \int_{-\infty}^{+\infty}p_1(x)\log\Big(\frac{p_1(x)}{p_2(x)}\Big)dx$ \cite{cover1999elements}. Thus, $\int_{-\infty}^{+\infty}p_1(x)\log\Big(\frac{p_1(x)}{p_1(x)+p_2(x)}\Big)dx \propto  KL(p_1||(p_1+p_2)/2)$. Here, $(p_1+p_2)/2$ is used to ensure validity of the probability distribution such that $\int_{x} (p_1(x)+p_2(x))/2) = 1$. 
Substituting into Eqn. (\ref{Eq:IntermediateStep}), we obtain
\begin{align}
V_{\text{min}} = &\text{KL}\Big({q}_{T} || \frac{{q}_{T} + {q}_{C}}{2}\Big) + \text{KL}\Big({q}_{C} || \frac{{q}_{T} + {q}_{C}}{2}\Big)
\label{eq:Gval1} \\
    &+
    \mathbb{E}_{x \sim {p}} [\ell_{\theta_T} (x \times (1-\Delta)+ \delta \times \Delta, y_T) + \ell_{\theta_T}(x, y_C)]. \nonumber
\end{align}
Recall from Sec. \ref{sec:GameModel} that $q_T$ and $q_C$ are output distributions from clean and Trojaned DNNs when inputs $x$ to the models are generated from a multi-variate Gaussian distribution $\mathcal{N}$.
Hence the first two terms in Eqn.~\eqref{eq:Gval1} depend only on input samples from the set $\mathscr{D}_R = \{x~|~x \sim \mathcal{N}\}$ while the last two terms depend only on input samples from $\mathscr{D} = \{(x,y)~|~(x,y) \sim p\}$.




The proof follows by observing that the first two terms in Eqn. (\ref{eq:Gval1}) become zero when $q_T =\frac{q_{T} + q_{C}}{2}$, i.e., $q_T = q_C$.
\end{proof}

\subsection{Analysis of Test Loss Function}\label{App:sub}
In this section, we prove that when the loss function for selecting input samples for trigger-embedding is the cross-entropy loss, the adversary's test loss function satisfies a diminishing returns or supermodularity property. 
We first present a formal definition of supermodularity.

\begin{defn}[\cite{fujishige2005submodular}]\label{defn:sub}
Given a finite set $\Omega$, a function $g: 2^{\Omega} \rightarrow \mathbb{R}$ is submodular if 
$g( V \cup \{v\}) - g(V) \geq g(U \cup \{v\}) - g(U)$ 
for any $V \subseteq U \subseteq \Omega$ and any $v \in \Omega \setminus U$.
Function $g$ is said to be supermodular if $-g$ is submodular.
\end{defn}

Given two sets $V$ and $U$ such that $V \subseteq U \subseteq \Omega$, Definition \ref{defn:sub} indicates that the marginal increase in the value of the function $g$ is higher when an element $v \in \Omega \setminus U$ is added to a smaller subset $V$ than when the same sample is added to a larger subset $U \supseteq V$. 
For any $U,V\subseteq\Omega$, we say a function $g: 2^{\Omega} \rightarrow \mathbb{R}$ is weakly submodular if $\sum_{u\in U}(g(V\cup \{u\})-g(V))\geq \gamma (g(V\cup U)-g(V))$ for some $\gamma\in(0,1]$.
To show such supermodularity property for the test-loss function, we first present a result from \cite{killamsetty2021glister} which holds for a class of commonly used loss-functions.

\begin{lem}[\cite{killamsetty2021glister}]\label{lemma:prelim}
    When $\ell_{\theta}$ is chosen as negative logistic loss, negative squared loss, negative hinge loss, perceptron loss, the data subset selection problem is an instance of submodular maximization problem. When $\ell_{\theta}$ is the negative cross entropy loss, the data subset selection problem is an instance of weakly submodular maximization problem.
\end{lem}

We recall notation from Sec. \ref{sec:Submod}. 
The set of clean samples is denoted by $\mathscr{D}$ and the subset of clean samples into which a trigger is embedded by $\tilde{\mathscr{D}}$. 
In order to evaluate the efficacy of the set of trigger-embedded inputs in training a Trojaned DNN, the adversary will use a test dataset, denoted  $\mathscr{D}_{\text{Test}}$, where $\mathscr{D}_{\text{Test}}$ contains samples that are not part of the training dataset. 
Additionally, the adversary will use a test set, denoted $\mathscr{\Tilde{D}}_{\text{Test}}$, for embedding a Trojan trigger. 
We use $\mathcal{L}_{T}(\mathscr{D},\mathscr{\Tilde{D}})$ and $\mathcal{L}_{C}(\mathscr{D},\mathscr{\Tilde{D}})$ to represent loss functions computed using  $\mathscr{\Tilde{D}}_{\text{Test}}$ and $\mathscr{D}_{\text{Test}}$, and $\mathcal{L}_{tot}(\mathscr{D},\mathscr{\Tilde{D}}):=\mathcal{L}_{T}(\mathscr{D},\mathscr{\Tilde{D}}) + \mathcal{L}_{C}(\mathscr{D},\mathscr{\Tilde{D}})$ to denote the (total) {\em test loss}. 
Then, determining the subset of input samples into which to embed a Trojan trigger can be characterized as (also Eqn. (\ref{eq:Total_loss})): 
\begin{eqnarray*}
    &&\min_{\mathscr{\Tilde{D}} \subset \mathscr{D}} \mathcal{L}_{T}(\mathscr{D},\mathscr{\Tilde{D}}) + \mathcal{L}_{C}(\mathscr{D},\mathscr{\Tilde{D}}).
\end{eqnarray*}
%
In our case, the subset of data of interest is the subset of samples into which a trigger will be embedded, and the loss function will be $\mathcal{L}_{T}(\mathscr{D}, \mathscr{\Tilde{D}})$. Based on this choice of subset and loss function, we present a variant of Lemma \ref{lemma:prelim} below. 
\begin{lem}\label{lemma:supermodular}
    When $\ell_{\theta}$ is chosen as logistic loss, squared loss, hinge loss, perceptron loss, or cross entropy loss, the test loss-function term associated with predicting trigger-embedded samples to the adversary-desired target class, $\mathcal{L}_{T}(\mathscr{D}, \mathscr{\Tilde{D}})$ is (weakly) supermodular in the trigger-embedded dataset, $\mathscr{\Tilde{D}}$.
\end{lem}
Although Lemma \ref{lemma:supermodular} indicates that $\mathcal{L}_{T}(\mathscr{D}, \mathscr{\Tilde{D}})$ is (weakly) supermodular in the trigger-embedded dataset, $\mathscr{\Tilde{D}}$, it does not imply (weak) supermodularity of $\mathcal{L}_{tot}(\mathscr{D}, \mathscr{\Tilde{D}})$.
In the following, we show that under certain conditions, the total test-loss function in Eqn.~\eqref{eq:Total_loss} is supermodular.
\begin{theorem}\label{thm:supermodular}
Consider two datasets $\mathscr{\Tilde{D}}_1$ and $\mathscr{\Tilde{D}}_2$, where $\mathscr{\Tilde{D}}_1\subset \mathscr{\Tilde{D}}_2$, $|\mathscr{\Tilde{D}}|_1=s$, $|\mathscr{\Tilde{D}}|_2=k$, and $k>s>0$.
Let $m_C$ and $m_T$ be defined as
\begin{align*}
    m_C^s &= \mathcal{L}_C(\mathscr{D}, \mathscr{\Tilde{D}}_1) - \mathcal{L}_C(\mathscr{D}, \mathscr{\Tilde{D}}_1\cup\{(x,y)\}),\\
    m_C^k &= \mathcal{L}_C(\mathscr{D}, \mathscr{\Tilde{D}}_2) - \mathcal{L}_C(\mathscr{D}, \mathscr{\Tilde{D}}_2\cup\{(x,y)\}),\\
    m_T^s &= \mathcal{L}_T(\mathscr{D}, \mathscr{\Tilde{D}}_1) - \mathcal{L}_T(\mathscr{D}, \mathscr{\Tilde{D}}_1\cup\{(x,y)\}),\\
    m_T^k &= \mathcal{L}_T(\mathscr{D}, \mathscr{\Tilde{D}}_2) - \mathcal{L}_T(\mathscr{D}, \mathscr{\Tilde{D}}_2\cup\{(x,y)\}),
\end{align*}
where $\{(x,y)\}\notin \mathscr{\Tilde{D}}_2$.
If $m_C^s+m_T^s\geq m_C^k+m_T^k$ for all $k>s>0$ and $\{(x,y)\}\notin \mathscr{\Tilde{D}}_2$, then the total test-loss function in Eqn. (\ref{eq:Total_loss}), $\mathcal{L}_{\text{tot}}(\mathscr{D},\mathscr{\Tilde{D}})$ is supermodular in $\mathscr{\tilde{D}}$.  
\end{theorem}


\begin{proof}
Our proof follows from the supermodularity property from Definition \ref{defn:sub}.
Using Eqn. \eqref{eq:Total_loss}, we have that
\begin{multline}
    \mathcal{L}_{tot}(\mathscr{D}, \mathscr{\Tilde{D}}_1) - \mathcal{L}_{tot}(\mathscr{D}, \mathscr{\Tilde{D}}_1\cup\{(x,y)\})\\
    - (\mathcal{L}_{tot}(\mathscr{D}, \mathscr{\Tilde{D}}_2) - \mathcal{L}_{tot}(\mathscr{D}, \mathscr{\Tilde{D}}_2\cup\{(x,y)\}))\\
    =m_C^s + m_T^s - (m_T^k+m_T^k).
\end{multline}
Since $m_C^s + m_T^s \geq  (m_T^k+m_T^k)$ by assumption, we have $\mathcal{L}_{tot}(\mathscr{D}, \mathscr{\Tilde{D}}_1) - \mathcal{L}_{tot}(\mathscr{D}, \mathscr{\Tilde{D}}_1\cup\{(x,y)\})- (\mathcal{L}_{tot}(\mathscr{D}, \mathscr{\Tilde{D}}_2) - \mathcal{L}_{tot}(\mathscr{D}, \mathscr{\Tilde{D}}_2\cup\{(x,y)\}))\geq 0$ for all $k>s>0$.
Therefore, from Def. \ref{defn:sub}, the total test-loss $\mathcal{L}_{tot}(\mathscr{D}, \mathscr{\Tilde{D}})$ is supermodular.
\end{proof}

The conditions in Theorem \ref{thm:supermodular} imply that the marginal impact on the total test-loss function by adding one (new) trigger-embedded sample into $\mathscr{\Tilde{D}}$ reduces as $|\mathscr{\Tilde{D}}|$ increases.

\newpage 
\end{document}